\documentclass[12pt]{article}

\usepackage{fullpage}
\usepackage{amssymb}
\usepackage[latin1]{inputenc}
\usepackage{graphicx}
\usepackage{float}
\usepackage{makeidx}
\usepackage{bm}
\makeindex
 \usepackage{amsmath} %\usepackage{amstext}
\usepackage{hyperref} % NOTE: this may cause you to have to hit compile a couple times initially.
\usepackage{xargs}
\usepackage{tikz}
\usetikzlibrary{matrix}
\usepackage{times}

\newenvironment{proof}{\noindent {\bf Proof:} \hspace{.4em}}
                      {\hspace{\fill}{$\blacksquare$} \smallskip}

\newenvironment{proofL}{\hspace{.4em}}
                      {\hspace{\fill}{$\blacksquare$} \smallskip}

\newtheorem{theorem}{Theorem}[subsection]

\newtheorem{lemma}[theorem]{Lemma}

\newtheorem{obs}[theorem]{Observation}
\newtheorem{definition}[theorem]{Definition}

\newtheorem{open}{Open Problem}

\newcommand{\troot}{t_{root}}
\newcommand{\troots}{{t_{root}}^*}

\newcommand{\tbss}{{t_b}^{**}}
\newcommand{\tbs}{{t_b}^{*}}
\newcommand{\trootsFirst}{{t_{root}}^{*First}}

\newcommand{\tbith}{t_b[{P_i}]}
\newcommand{\tblast}{t_b[{P_{last}}]}

\begin{document}

\thispagestyle{empty}

\vspace*{1.4in}

\begin{scshape}\begin{center}\fontsize{22}{22}\selectfont Fractional Pebbling Game Lower Bounds\end{center}

\vspace*{0.10in}

\begin{center}\noindent Frank Vanderzwet\end{center}

\vspace*{0.25in}

\begin{center}\noindent Advisors : Stephen Cook and Toniann Pitassi \end{center}

\vspace*{0.15in}

\begin{center}Department of Computer Science\\University of Toronto\end{center}

\vspace*{0.15in}

\begin{center}Draft of \today\end{center}\end{scshape}

\newpage

\tableofcontents

\newpage

\setcounter{section}{0}

\section{Introduction} 

%\subsection{Relation to Previous Work} 

The primary result in this paper involves the $fractional $ $pebbling$ $game$. The origins of the $fractional $ $pebbling$ $game$ is the $black$ $pebbling$ $game$. The $black$ $pebbling$ $game$ was introduced by Paterson and Hewitt \cite{ph:comsc} to compare the power of programming languages. Since this time, variations of the pebbling game have been used in many areas of Computer Science.

The pebbling game and related terms are more rigorously defined in Section 2. The definitions presented in this paper are reinterpretations of the definitions for pebbling games presented by Cook, McKenzie, Wehr, Braverman, and Santhanam \cite{c:pebjournal}. Surveys of the pebbling games are available, \cite{ppgr:survey} and \cite{nord:survey2}.

The game is played on DAGs. Each node in the DAG may have up to one pebble. Configurations are allocations of pebbles to nodes. There is one distinguished node. The goal is to reach a configuration that has a pebble on the distinguished node while the final configuration must end with no pebbles in the DAG. Configurations of pebbles are changed from one to another via the following moves : 

\begin{itemize}

\item
Place or remove a black pebble on a leaf node

\item
Place a black pebble on a node that has all children pebbled

\item
Remove a black pebble from a node

\end{itemize}

For the $black$ $pebbling$ $game$, lower and upper bounds for balanced trees are given in \\\cite{c:pebjournal}. Similar, motivational, lower and upper bounds are replicated in Section 3. 

Each node can be thought of as having a value and a method of determining that value from the value of its children. Black pebbles can be thought of as values deterministically computed from previous values. This analogy is essentially the $tree$ $evaluation$ $problem$ \cite{c:pebjournal}.\\

Branching programs are a nonuniform model of a Turing machines. Branching programs are directed multi-graphs whose nodes are states. Every edge is labelled with a value. There is one initial state from which the computation starts. Every state queries a variable and branches to new states along edges labelled with that value. These computations may eventually reach accepting or rejecting states.

A state in a branching program corresponds to a turing machine configuration. Thus  L $\ne$ P if we can show the branching programs solving a problem in P requires a superpolynomial number of states. 

With this goal, \cite{c:pebjournal} examined a restricted class of branching programs. A thrifty branching program for the $tree$ $evaluation$ $problem$ must query the value of the functions only at the correct value of the children. The thrifty hypothesis states that thrifty branching programs are optimal among all branching programs.

Under the thrifty hypothesis, $black$ $pebbling$ $game$ lower bounds allow for a proof of deterministic branching program lower bounds which separate L from P \cite{c:pebjournal}. It is hoped that $fractional $ $pebbling$ $game$ lower bounds allow for a similar proof for nondeterministic branching programs, which under the thrifty hypothesis, would separate NL from P.\\

Another variation of the pebbling game is the whole $black$-$white$ $pebbling$ $game$. It was introduced by Cook and Sethi \cite{cs:pfromnl} in an attempt to separate  NL and P. It is similar to the $black$ $pebbling$ $game$ except the rules for changing one configuration to another are the following :

\begin{itemize}

\item
Place or remove a pebble on a leaf node

\item
Place a black pebble on a node that has all children pebbled

\item
Remove a black pebble from a node

\item
Place a white pebble on a node

\item
Remove a white pebble from a node that has all children pebbled

\end{itemize}

White pebbles can be thought of as non-deterministic guesses for values. When we removed them we have essentially justified those guesses.\\

The pebbling games are important due to their relation to propositional proof complexity, particularly resolution. For this purpose, the whole $black$-$white$ $pebbling$ $game$ is usually used. Aspects of the game are encoded as a CNF formulas. Properties of the formulas are then argued based on properties of the pebbling game. \cite{nord:survey2} produced a survey of how the pebbling games relate to proof complexity.

Aleknovich showed a separation between regular and general resolution using a problem that is a modified version of the whole $black$-$white$ $pebbling$ $game$ \cite{alek:sepres}.

Using the pebbling contradiction problem derived from the pebbling game, Nordstrom showed resolution refutations of small widths may have large space requirements \cite{n:narrowpf}. Ben-Sasson showed, using the same pebbling contradictions, trade-offs between time size space and width of resolution \cite{bs:wsstrade}.\\

Motivated by proving lower bounds for branching programs \cite{c:pebjournal} recently introduced the $fractional$ $pebbling$ $game$. The $fractional $ $pebbling$ $game$ is a generalization of the whole $black$-$white$ $pebbling$ $game$. 

The rules are similar to those presented in the whole $black$-$white$ $pebbling$ $game$ except we now allow for fraction of pebbles. 

The $fractional $ $pebbling$ $game$ should better represent the non-deterministic approach to the problem than the whole $black$-$white$ $pebbling$ $game$. Fractions of pebbles can be thought of as partially specifying the possible values of a node. This intuitively is helpful and seems less restrictive than the whole $black$-$white$ $pebbling$ $game$. We confirm that this is helpful by showing smaller lower bounds for the $fractional $ $pebbling$ $game$ than are possible for the whole $black$-$white$ $pebbling$ $game$. These lower bounds match upper bounds presented in \cite{c:pebjournal} for the $fractional $ $pebbling$ $game$.

The main theorem we show in this paper (Section 5.3) relies on the $fractional$ $pebbling$ $game$. Let  $T^h_d$ be the balanced $d$-ary tree of height $h$. Let $min_h = (d-1) * h/2+1$. Let the root node be the node that must be pebbled. \\

\noindent
{\bf Main Theorem}

\noindent
In every fractional pebbling of $T^h_d$, where the distinguished node is the root, there is a configuration such
that the number of pebbles is greater than or equal to $min_h$.\\

Loose lower bounds for this problem were presented in \cite{c:pebjournal} and tight lower bounds were left as an open problem. In that case the lower bounds for the problem came from a reduction to a paper by Klawe \cite{k:bwpyr} which proves the bounds for pyramid graphs rather than balanced trees. Accuracy is lost in the reduction. We present tight lower bounds for balanced trees of any degree by taking a more direct approach.

We will solve this problem using a $shifting$ argument. The idea in our shifting argument is that if we use less pebbles before placing a pebble on the root we use more pebbles after placing a pebble on the root. We proceed in this manner since we must cover a larger range of pebbling strategies once we allow for fractional pebbles.\\

\subsection{Organization} 

The organization of this paper is as follows. Section 2 defines the pebbling game and associated terms. It first defines the $black$ $pebbling$ $game$ and the whole $black$-$white$ $pebbling$ $game$. It then defines the $half$ $pebbling$ $game$ and the $fractional $ $pebbling$ $game$ as modifications of the whole $black$-$white$ $pebbling$ $game$. Further we define terms related to all games. Section 3 first demonstrates upper bounds for the $black$ $pebbling$ $game$. It then demonstrates lower bounds for the $black$ $pebbling$ $game$. In Section 4 we show upper and lower bounds for the whole $black$-$white$ $pebbling$ $game$. Section 5 shows upper bounds for the half and $fractional $ $pebbling$ $game$s and concludes by showing $fractional $ $pebbling$ $game$ lower bounds.

\newpage

\section{Preliminaries} 

In Section 3 we examine the $black$ $pebbling$ $game$. We next present definitions and rules needed for the $black$ $pebbling$ $game$ played on DAGs.

\begin{definition}

A {\bf black pebble configuration} on a DAG is an
assignment of values b(i) to each node i of the tree, where\\
b(i) = 0 or b(i) = 1\\
We let b(i) represent the {\bf black pebble weight value} of i.
\end{definition}

\begin{definition}
\noindent
A {\bf black pebble move} changes one black pebble configuration into another. Possible black pebble moves are :\\
(i) For any node i, decrease b(i) from 1 to 0\\
(ii) For any node i, if each child of i has pebble value 1, increase b(i) to 1, and optionally decrease any of the black pebble values of the children of i to 0\\
(iii) For each leaf node i, increase b(i) to 1
\end{definition}

For (ii), if we choose to decrease the black pebble value of the children it is done simultaneously, this is called a {\bf black sliding move}.

\begin{definition}
\noindent
A {\bf black pebbling} $\pi$ is a sequence $m_1,m_2,\ldots$
of black pebble moves resulting in a sequence $c_0,c_1,c_2,\ldots$,
of black pebble configurations, where $c_0$ is the initial configuration,
and for $t>0$, $c_t$ is the configuration after move $m_t$.\\
\end{definition}

We next present definitions needed for the $whole$ $black$-$white$ $pebbling$ $game$. Upper bounds for this game are presented in Section 4.

\begin{definition}

A {\bf whole black-white pebble configuration} on a DAG, is an
assignment of a pair of numbers (b(i),w(i)) to each node i of the tree, where\\
b(i) = 0 or b(i) = 1,\\
w(i) = 0 or w(i) = 1 and\\
b(i) + w(i) $\leq$ 1\\
Here b(i) and w(i) are the {\bf  black pebble weight value} and the {\bf white pebble weight value}, respectively, of node i,
and b(i) + w(i) is the {\bf pebble weight} of node i.
\end{definition}

\begin{definition}
\noindent
A {\bf whole black-white pebbling move} changes one whole black-white pebble configuration into another. Possible whole black-white pebble moves are :\\
(i) For any node i, set b(i) to 0\\
(ii) For any node i, if each child of i has pebble value 1, set w(i) to 0, increase b(i) to 1, and optionally decrease any of the black pebble weight values of the children of i to 0\\
(iii) For any node i, increase w(i) to 1\\
(iiii) For each leaf node i, increase b(i) to 1
\end{definition}

\begin{definition}
\noindent
A {\bf whole black-white pebbling} $\pi$ is a sequence $m_1,m_2,\ldots$
of whole black-white pebble moves resulting in a sequence $c_0,c_1,c_2,\ldots$,
of whole black-white pebble configurations, where $c_0$ is the initial configuration,
and for $t>0$, $c_t$ is the configuration after move $m_t$.\\
\end{definition}

In Section 5.1 we use a variation of the whole $black$-$white$ $pebbling$ $game$ wherein we additionally allow b(i) and w(i) to be 0.5. We call this variation the $half$ $pebbling$ $game$. This closely resembles the $fractional $ $pebbling$ $game$ defined next.\\

In Section 5.2 and 5.3 we use a variation of the whole $black$-$white$ $pebbling$ $game$ that allows b(i) and w(i) to be any real number in [0,1]. We call this variation the $fractional$ $pebbling$ $game$:

\begin{definition}

A {\bf fractional pebble configuration} on a DAG, is an
assignment of a pair of real numbers (b(i),w(i)) to each node i of the tree, where\\
0 $\leq$ b(i),w(i) and\\
b(i) + w(i) $\leq$ 1\\
Here b(i) and w(i) are the {\bf  black pebble weight value} and the {\bf white pebble weight value}, respectively, of node i,
and b(i) + w(i) is the {\bf pebble weight} of node i.
\end{definition}

\begin{definition}
\noindent
A {\bf fractional pebble move} changes one fractional pebble configuration into another. Possible fractional pebble moves are :\\
(i) For any node i, decrease b(i) arbitrarily\\
(ii) For any node i, if each child of i has pebble value 1, decrease w(i) to 0, increase b(i) arbitrarily, and optionally decrease the black pebble weight values of the children of i arbitrarily\\
(iii) For any node i, increase w(i) such that b(i) + w(i) = 1\\
(iiii) For each leaf node i, increase b(i) arbitrarily 
\end{definition}

\begin{definition}
\noindent
A {\bf fractional pebbling} $\pi$ is a sequence $m_1,m_2,\ldots$
of fractional pebble moves resulting in a sequence $c_0,c_1,c_2,\ldots$,
of fractional pebble configurations, where $c_0$ is the initial configuration,
and for $t>0$, $c_t$ is the configuration after move $m_t$.\\
\end{definition}

We additionally define the following terms and symbols important to all variations of the games.

\begin{definition}
We refer to a configuration $c_t$ as the {\bf time} t.
\end{definition}

\begin{definition}
We let {\bf 0} denote the initial configuration, equivalently the {\bf initial time}.
\end{definition}

\begin{definition}
The {\bf weight}, $w_\pi(t)$, of $\pi$ at time $t$ is sum of the
pebble weights on $T$ in configuration $c_t$. The {\bf subtree weight}, $sw_\pi(t)$, of $\pi$ at time $t$ is the sum of the pebble weights in the principal subtrees of $T$ in configuration $c_t$. 
The {\bf white subtree weight}, $w.sw_\pi(t)$, of $\pi$ at time $t$ is the sum of the white pebble weights in the principal subtrees of $T$ in configuration $c_t$. The {\bf black subtree weight}, $b.sw_\pi(t)$, of $\pi$ at time $t$ is the sum of the black pebble weights in the principal subtrees of $T$ in configuration $c_t$. The {\bf root weight}, $rw_\pi(t)$, of $\pi$ at time $t$ is the pebble weight on the root of $T$ in configuration $c_t$. The {\bf black root weight}, $b.rw_\pi(t)$, of $\pi$ at time $t$ is the black pebble weight on the root of $T$ in configuration $c_t$. The {\bf white root weight}, $w.rw_\pi(t)$, of $\pi$ at time $t$ is the white pebble weight on the root $T$ in configuration $c_t$.\\

\noindent
Square brackets after the symbols defined above are used to indicate in which tree or subtree the pebble weight is located. For example, the symbol $b.rw_\pi(t)$[$P_{last}$] would be used to specify some amount of black pebble weight on the root of the tree $P_{last}$ at time t. If it is not specified, the symbol is assumed to pertain to the entire tree.
\end{definition}

\begin{definition}
\noindent
A {\bf root-pebbling} is a pebbling that requires that the initial and final pebble weights of $\pi$ are 0, and 
$rw_\pi(t)=1$ at some time $t$. 

\noindent
A {\bf sub-pebbling} is a pebbling that may start or end with pebble weight. It may initially have arbitrary white pebble weight and at the end of the pebbling it may have arbitrary black pebble weight. It may also have some specified initial black pebble weight. At the end of the pebbling it has no white pebble weight.

\noindent
A {\bf root sub-pebbling} is a sub-pebbling such that $rw_\pi(t)=1$ at some time $t$. 

\noindent
Similarly, a {\bf sub-root sub-pebbling} is a sub-pebbling such that the subtrees of $T$ have $rw_\pi$(t)=1 at some time $t$.
\end{definition}

\begin{lemma}
If $\pi_1$ is a sub-pebbling with initial white and black pebble weight, and $w_{\pi_1}(t) \leq P$ for all times $t$ then there exists a sub-pebbling $\pi_2$ with the same initial black pebble weight and no white pebble weight such that $w_{\pi_2}(t) \leq P$ for all times $t$.
\end{lemma}

\begin{proofL}
We show such a $\pi_2$. The first steps is to place the same white pebble weight on the same nodes as initially in $\pi_1$. We then could follow the sub-pebbling $\pi_1$. Since we have less pebble weight before we add the white pebble weight, $w_{\pi_2}(t) \leq P$ for all times t.
\end{proofL}\\

This lemma indicates that initial white pebble weight is not helpful.\\

In all pebbling games we allow for a {\bf black sliding} move. This is pebble move (ii) in all games. Rule (ii) is sometimes alternatively written as follows :

\noindent
(ii) For any node i, if each child of i has pebble value 1, increase b(i) arbitrarily.

This would be the case if we did not allow for black sliding moves. This decouples increasing pebble weight and removing pebble weight from the children.

\begin{obs}
A pebbling with {\bf black sliding} moves can be converted to a pebbling without black sliding moves which requires at most 1 more pebble weight.
\end{obs}

This is simply the result of changing a black sliding move to two subsequent moves. We allow {\bf sliding} moves in our proofs.

\begin{definition}
We let $T^h_d$ represent the balanced $d$-ary tree of height $h$.
\end{definition}

\newpage

\section{Black Pebbling Game} 

\subsection{Black Pebbling Game Upper Bounds} 

We prove the following theorem which shows an upper bound for the $black$ $pebbling$ $game$ defined in Section 2. Similar results can be found in \cite{c:pebjournal}.

\begin{theorem}
Let $min_h = h$.
There exists a $black$ $pebbling$ $game$ {\bf root-pebbling} $\pi$ of $T^h_2$, $h \geq 2$, such
that for all times t, $w_\pi(t) \le min_h$.
\end{theorem}

\noindent
To show this we use induction.\\

%\noindent
%{\bf Induction Hypothesis [IH$(h)$]:} 

%\noindent
%Let $min_h = h$. For $h\geq 2$
%(where $b.sw_{\pi}$(0) is the initial black subtree pebble weight) 
%there exist a $black$ $pebbling$ $game$ {\bf root-pebbling} $\pi$ of $T_2^h$ such that $sw_\pi \leq min_h$ at all times.\\

\noindent
{\bf Base Case :} h = 2.

There are 2 children of the root. We place a black pebble weight on each leaf and slide a black pebble weight to the root. Thus, $sw_\pi \leq 2$ at this time and all previous times. Thus the IH is satisfied in the base case.\\

%\noindent
%{\bf Base Case :} h = 3.

%There are $d$ subtrees of the root. On the first subtree, place $d$ pebbles on the leaves of the given subtree and slide a pebble to that subtrees root. We then remove black pebble weight that is not on the root. We do the same in each subtree until we reach the final subtree. In the final subtree we now use d pebble weight on the leaves of the subtree. At this time we have pebble weight d in this subtree while maintaining one pebble weight in each of the other ($d-1$) subtrees. Thus, $sw_\pi \leq d + (d - 1)$ at this time and all previous times.

%We next slide a pebble to the root of this final subtree. We use the same pebble weight at this time as at the previous time since we have slid a pebble. We now have a pebble on each subtree root and slide a pebble to the root. Thus we pebble the root and $sw_\pi \leq d + (d - 1) = min_h$ at all times.

%Thus the IH is satisfied in the base case.\\

\noindent
{\bf Induction step :} We prove for $h+1$ assuming for $h'$,
$3\le h' \le h$.\\
Note $min_{h+1} = min_h + 1$.\\

There are two subtrees of the root. Using $min_h$ pebble weight we pebble the first subtree root using the pebbling in the IH for height $h$. We then remove black pebble weight that is not on the root of the subtree such that we have only this 1 pebble weight. 

We next use $min_h$ pebble weight to pebble the second subtree root using the pebbling in the IH for height h. At this time we maintain one pebble weight in the first subtree. We thus use $sw_\pi \leq min_h + 1$.

We now have a pebble on each subtree root and slide a pebble to the root. Thus, $sw_\pi \leq min_h + 1 = min_{h+1}$ at all times.

Thus, the IH is satisfied.\\

To show this for $d$-ary balanced trees we would iteratively pebble the children of the root using the pebbling in the IH. Each time leaving a pebble. This would result in an upper bound of $(d-1) * (h-1) + 1$.

The key insight is that we had to leave some pebble weight in one subtree while we proceeded with the pebbling in another subtree. This idea is important to all subsequent proofs.\\\\

\subsection{Black Pebbling Game Lower Bounds}

We prove the following theorem which shows a lower bound for the $black$ $pebbling$ $game$ defined in Section 2. Combined with the previous section we have a tight bound on the number of pebbles taken to complete the $black$ $pebbling$ $game$ for balanced trees of degree 2. Similar results have been shown in \cite{c:pebjournal}.\\

\begin{theorem}\label{bplb}
Let $min_h = h$.
For every $black$ $pebbling$ $game$ {\bf root-pebbling} $\pi$ of $T^h_2$, $h \geq 2$, there is a time $t$ such
that $w_\pi(t) \ge min_h$.
\end{theorem}

\noindent
To show this we use induction.\\

%\noindent
%{\bf Induction Hypothesis [IH$(h)$]:} 

%\noindent
%Let $min_h = h$. If $h\geq 2$
%and $\pi$ is a $black$ $pebbling$ $game$ {\bf root-pebbling} of $T_2^h$, then there is a time t such that $sw_\pi(t) \geq min_h$.\\

\noindent
{\bf Base Case :} h = 2.

There are 2 children of the root. To place a pebble on the root we must pebble these 2 nodes. The IH is then satisfied in the base case.\\

\noindent
{\bf Induction step :} We prove for $h+1$ assuming for $h'$,
$3\le h' \le h$.\\
Note $min_{h+1}$ = $min_h + 1$.\\

There are 2 subtrees of the root. There must be a time before we pebble the root that we have a pebble on each subtree root if we are to place a pebble on the root. Thus, by IH, there must be a last time we use pebble weight $min_h$ in one of the subtrees. Let this time be $t_{last}$. 

At $t_{last}$, suppose for contradiction we did not have one pebble in the other subtree.  Having less than one black pebble on any node does not allow us to apply any of the pebbling rules and is thus equivalent to having no pebble weight.

To pebble the root we must have a pebble on each of the subtree roots. Thus if we had less than one pebble in any subtree we must place a pebble on the root of that subtree before we pebble the root. To do this we require $min_h$ pebble weight by IH. This would contradict $t_{last}$ being the last time we use pebble weight $min_h$.

Thus we maintain at least one pebble in the other subtree at $t_{last}$ and $w_\pi(t_{last}) \geq min_h + 1 = min_{h + 1}$ as required.

Thus, the IH is satisfied.\\

To show this for $d$-ary balanced trees we would look at the last time we use $min_h$ in any tree and argue that we need 1 pebble in each other subtree at this time. 

This would result in an lower bound of $(d-1) * (h-1) + 1$. The proofs in this section result in a tight lower bound for the $black$ $pebbling$ $game$ on balanced binary trees. We will show a tight lower bound for the $fractional $ $pebbling$ $game$.

\newpage

\section{Whole Black-White Pebbling Game} 

\subsection{Whole Black-White Pebbling Game Upper Bounds} 

We prove the following theorem which shows an upper bound for the whole $black$-$white$ $pebbling$ $game$ defined in Section 2. Similar results can be found in \cite{c:pebjournal}.

\begin{theorem}
Let  $min_h = \lceil h/2\rceil + 1$.
There exists a whole $black$-$white$ $pebbling$ $game$ {\bf root-pebbling} $\pi$ of $T^h_2$, $h \geq 2$, such
that for all times t, $w_\pi(t) \le min_h$.
\end{theorem}

\noindent
To show this we use induction. We show this only for the even height cases and it follows for the odd height cases since we can extract a pebbling for an odd height from the larger even height pebbling.\\

\noindent
{\bf Induction Hypothesis [IH$(h)$]:} 

\noindent
Let $min_h = h/2 + 1$.\\
For even $h\geq 2$
there exist a $whole$ $black$-$white$ $pebbling$ $game$ {\bf root-pebbling} $\pi$ of $T_2^h$ and a time $\troot$ such that $sw_\pi \leq min_h$ at all times. Additionally,

(1) $b.rw_\pi(\troot) = 1$

(2) $w.w_\pi(\troot) \leq min_h-2$

(3) White pebble weight at $\troot$ can be removed using $w_\pi(t) \leq min_h$ for  $t > \troot$\\

Condition (1) specifies that the root weight at $\troot$ is black. Condition (2) specifies that there is not too much white pebble weight at $\troot$.\\

\noindent
{\bf Base Case :} h = 2.

There are 2 children of the root. We use 2 pebble weight on the leaves and slide it to the root. Thus, $sw_\pi \leq 2$ at this time and all previous times. Condition (2) and (3) are satisfied since we have no white pebble weight. Thus the IH is satisfied in the base case.\\

\noindent
{\bf Induction step :} We prove the induction hypothesis for $h+2$ assuming it for $h'$,
$2\le h' \le h$.\\
Note $min_{h+2}$ = $min_h + 1$.\\

We let the children of the root be $p_2$ and $p_3$. We call the children of these $v_1$, $v_2$, $v_3$ and $v_4$ as in the following figure.\\ 

\begin{figure}[H]
  \centering
%\%includegraphics{./thesisimgs/lemma0caseII1.png}
%\addsomegraphic{./thesisimgs/lemma0caseII1.png}{0cm}{0cm}
%\flushgraphics

\begin{tikzpicture}\label{top3graph}

\draw [-] (-3,0.1) -- (10,0)[color=white] node [below] {};
\draw [-] (1,0.1) -- (2,0.5) node [below] {};
\draw [-] (2.75,0.1) -- (2.25,0.5) node [below] {};
\draw [-] (4.25,0.1) -- (4.75,0.5) node [below] {};
\draw [-] (6,0.1) -- (5,0.5) node [below] {};
\draw [-] (2.1, 0.9) -- (3.2,1.3) node [below] {};
\draw [-] (4.85, 0.9) -- (3.85,1.3) node [below] {};

\draw [-] (1,-0.3) -- (1.25,-0.5) node [below] {};
\draw [-] (2.75,-0.3) -- (3,-0.5) node [below] {};
\draw [-] (4.25,-0.3) -- (4.5,-0.5) node [below] {};
\draw [-] (6,-0.3) -- (6.25,-0.5) node [below] {};
\draw [-] (1,-0.3) -- (0.75,-0.5) node [below] {};
\draw [-] (2.75,-0.3) -- (2.5,-0.5) node [below] {};
\draw [-] (4.25,-0.3) -- (4,-0.5) node [below] {};
\draw [-] (6,-0.3) -- (5.75,-0.5) node [below] {};

    \node[text width=4 cm,text ragged, text height = 5, anchor=west]  at (3.05,1.55) {root};
    \node[text width=4 cm,text ragged, text height = 5, anchor=west]  at (1.85,0.8) {$p_2$};
    \node[text width=4 cm,text ragged, text height = 5, anchor=west]  at (4.6,0.8) {$p_3$};
    
    \node[text width=4 cm,text ragged, text height = 5, anchor=west]  at (0.7,-0.06) {$v_1$};
    \node[text width=4 cm,text ragged, text height = 5, anchor=west]  at (2.45,-0.06) {$v_2$};
    \node[text width=4 cm,text ragged, text height = 5, anchor=west]  at (3.95,-0.06) {$v_3$};
    \node[text width=4 cm,text ragged, text height = 5, anchor=west]  at (5.7,-0.06) {$v_4$};
\end{tikzpicture}

\caption{Our labeling of the nodes of $T_2^h$.}
\end{figure}
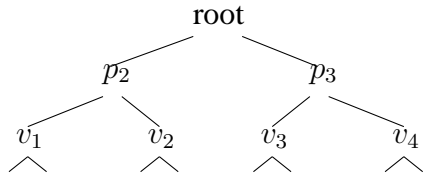

We simulate the pebbling in the IH for height $h$ in the subtree rooted at $v_1$. We modify the pebbling to leave the pebble on the root. This requires at most $min_h+1$ pebble weight.

We then simulate the pebbling in the subtree rooted at $v_2$. We interrupt the pebbling when $v_2$ is pebbled. We use $sw_\pi \leq min_h + 1$ at all times before this point.

We remove all other black pebble weight in $v_2$ such that we have $min_h - 2$ white pebble weight in the subtree rooted at $v_2$ by condition (2) and an additional pebble on $v_2$. At this time we maintain one pebble weight on $v_1$. We then have $min_h$ pebble weight in the tree.

We then slide a pebble to $p_2$. We then place a white pebble on $p_3$. We may then slide a pebble to the root. At this point we have 1 pebble on the root, 1 pebble on $p_3$ and $min_h - 2$ in the subtree rooted at $v_2$. By sliding the pebble to the root we satisfy condition (1). 

We then remove all black pebble weight and have white pebble weight $min_h - 1$, satisfying (2). We have yet to exceed $min_{h+2}$. We only have white pebble weight present at $\troot$ thus removing it will show (3).\\

We remove the $min_h - 2$ white pebble weight in the subtree rooted at $v_2$. This takes $min_h$ by condition (3) of the IH. The only other pebble weight is on $p_3$. Thus condition (3) has yet to be violated and we still have not exceeded $min_{h+2}$.

We simulate the pebbling in the subtree rooted at $v_3$ and interrupt it when there is a pebble on $v_3$. We remove all black pebble weight other than on the $v_3$. At this point there is 1 pebble on $p_3$, 1 pebble on $v_3$, and $min_h - 2$ white pebbles in the subtree rooted at $v_3$. We then place a white pebble on $v_4$. Thus we have yet to exceed $min_{h+2}$.

We remove the pebble on $p_3$ and the black pebble on $v_3$. We then remove the white pebble weight in the subtree rooted at $v_3$ using (3) from the IH.

To remove the white pebble on $v_4$ we simulate the pebbling for $h$ but remove the white pebble instead of placing a black pebble. We remove the resulting white pebble weight and the pebbling is complete. At no point in removing the white pebble weight that was present at $\troot$ have we used more that $min_h+1$, thus condition (3) and the IH are satisfied.\\\\

This shows the power of white pebbles. We next show that the upper bound for fractional pebbling can be obtained using only half pebbles. However, in Section 5 we show that fractional pebbles allow for a multitude of pebbling strategies. \\\\

\subsection{Whole Black-White Pebbling Game Lower Bounds} 

We prove the following theorem which shows a lower bound for the whole $black$-$white$ $pebbling$ $game$ defined in Section 2. Combined with the previous section we have a tight bound on the number of pebbles taken to complete the whole $black$-$white$ $pebbling$ $game$ for balanced trees of degree 2. Similar results have been shown in \cite{c:pebjournal}.\\

\begin{theorem}
Let $min_h = \lceil h/2 \rceil + 1$.
For every whole $black$-$white$ $pebbling$ $game$ {\bf root-pebbling} $\pi$ of $T^h_2$, $h \geq 2$, there is a time $t$ such
that $w_\pi(t) \ge min_h$.
\end{theorem}

\noindent
We show this by induction : \\

\noindent
{\bf Base Case :} h = 2

We must show that for $h=2$, if $\pi$ is a $whole$ $black$-$white$ $pebbling$ $game$ {\bf root-pebbling} of $T_2^2$, then there is a time t such that $sw_\pi(t) \geq 2$. This is trivially true.\\

\noindent
{\bf Base Case :} h = 3

We need to show that if $\pi$ is a $whole$ $black$-$white$ $pebbling$ $game$ {\bf root-pebbling} of $T_2^3$, then there is a time t such that $sw_\pi(t) \geq 3$.
 
If we ever use a white pebble we must use at least 3 pebbles at the time before we remove it. Thus we may not use white pebbles if we wish to use less than pebble weight 3.

Then, if we used less than 3 pebble weight, we would contradict Theorem \ref{bplb}.\\

\noindent
{\bf Induction step :}  Assuming the theorem is true for  $h'$,
$2 \le h' \le h$, it is sufficient to prove the following.

\begin{lemma} \label{bwlb3}
For $h\geq2$, if $\pi$ is a $whole$ $black$-$white$ $pebbling$ $game$ {\bf root-pebbling} of $T_2^{h+2}$ then there is a time t such that $sw_\pi(t) \geq min_{h+2}$.
\end{lemma}

\begin{proof} 

Note $min_{h+2} = min_h + 1$. For the sake of contradiction, suppose $sw_\pi(t) < min_{h} + 1$ or equivalently $sw_\pi(t) \leq min_{h}$ for all times t.
 
Since there is a time where the root is pebbled there must be a time where the children of the root are pebbled to add black pebble weight or to remove white pebble weight from the root. Let $\troots$ be a time such that $rw_\pi$($\troots$)=1 for both principal subtrees.
 
By the same logic we must pebble $v_1$, $v_2$, $v_3$ and $v_4$ (Figure 1). Thus, by the IH, it is the case that at some time we must use $min_h$ pebble weight in the subtrees rooted at these nodes. Note there may be more than one time fitting this description for each tree rooted at the $v_i$.

If two or more of these times occur before $\troots$ then at the last time there must be no pebble weight elsewhere in the tree. Thus we must again use $min_h$ in the subtrees that are not the last subtree.  Thus we will need to use $min_h$ in at least three subtrees after $\troots$ (subtrees rooted at $v_1$, $v_2$, $v_3$ or $v_4$).

When we use $min_h$ in the first such subtree after $\troots$ there can be no pebbles elsewhere. This would indicate we no longer need to reach such a time in any other subtrees. This is a contradiction since we need to use $min_h$ in at least three subtrees after $\troots$. Thus at some some time t, $sw_\pi(t) > min_{h}$ and $sw_\pi(t) \geq min_{h+2}$ as desired.\\
 \end{proof} 
 
 The previous proof is much simpler than the proof of the main theorem we will show later. This is due to the limited number of strategies possible when using whole pebbles.

\newpage

\section{Fractional Pebbling Game}

\subsection{Half Pebbling Game Upper Bounds} 

We prove the following theorem which shows an upper bound for the $half$ $pebbling$ $game$ defined in Section 2. Similar results can be found in \cite{c:pebjournal}.

\begin{theorem}
Let $min_h = h/2+1$.
There exists a $half$ $pebbling$ $game$ {\bf root-pebbling} $\pi$ of $T^h_2$, $h \geq 2$, such
that for all times t, $w_\pi(t) \le min_h$.
\end{theorem}

\noindent
To show this we use induction.\\

\noindent
{\bf Induction Hypothesis [IH$(h)$]:} 

\noindent
Let $min_h = h/2 + 1$.  Let $\troot$ be a time such that $rw_\pi(\troot)=1$.\\
For $h\geq 2$
there exist a $half$ $pebbling$ $game$ {\bf root-pebbling} $\pi$ of $T_2^h$ such that $sw_\pi \leq min_h$ at all times. Additionally,

(1) $b.rw_\pi(\troot) = 1$

(2) $w.w_\pi(\troot) \leq min_h-2$

(3) White pebble weight at $\troot$ can be removed using $w_\pi(t) \leq min_h$ for  $t > \troot$\\

\noindent
{\bf Base Case :} h = 2.

There are 2 children of the root. We place 2 black pebble weight on the leaves and slide it to the root. Thus, $sw_\pi \leq 2$ at this time and all previous times. Condition (2) and (3) are satisfied since we have no white pebble weight. Thus the IH is satisfied in the base case.\\

\noindent
{\bf Induction step :} We prove the induction hypothesis for $h+1$ assuming it for $h'$,
$2\le h' \le h$. Let $P_2$ and $P_3$ be the principal subtrees. Note $min_{h+1} = min_h + 0.5$.\\

We simulate the pebbling in the IH for height $h$ in $P_2$. We modify the pebbling to leave half a black pebble on the root. This requires at most half a pebble more or $min_{h+1}$ pebble weight.

We then simulate the pebbling in the IH for height $h$ in $P_3$. We interrupt the pebbling when the root of $P_3$ is pebbled. We use $sw_\pi \leq min_{h+1}$ at all times before this point.

We remove all other black pebble weight in $P_3$ such that we have $min_h - 2$ white pebble weight in the subtree $P_3$ by condition (2) and an additional pebble on the root of $P_3$.

We next add half a white pebble to the root of $P_2$ and slide a pebble from the root of $P_3$ to the root. Thus condition (1) is satisfied. We remove all black pebble weight and have half a white pebble on the root of $P_2$ and $min_h-2$ white pebble weight in $P_3$. We thus satisfy condition (2). Additionally, we only have white pebble weight present at this $\troot$ and removing it will show condition (3).\\

We remove the $min_h - 2$ white pebble weight in $P_3$. This takes $min_h$ pebble weight by condition (3) of the IH. The only other pebble weight is the half pebble on the root of $P_2$.

We simulate the pebbling from the IH for height $h$ in $P_2$. Instead of placing a black pebble we remove the white pebble on the root. This takes $min_h$ while maintaining the half a white pebble on the root of $P_2$. Thus condition (3) is not violated.

Thus the IH is satisfied.\\

We next show we can not do better using fractional pebbles. However, we also show there are strategies not available using only half pebbles.\\\\

\subsection{Fractional Pebbling Game Upper Bounds}

We prove the following theorem which shows an upper bound for the $fractional $ $pebbling$ $game$ defined in Section 2. Similar results have been known since \cite{c:pebjournal}.

\begin{theorem}
Let $min_h = (d-1) * h/2+1$.
There exists a $fractional $ $pebbling$ $game$ {\bf root-pebbling} $\pi$ of $T^h_d$ such
that for all times t, $w_\pi(t) \le min_h$.
\end{theorem}

\noindent
To show this we use induction.\\

\noindent
{\bf Induction Hypothesis [IH$(h)$]:}

\noindent
Let $min_h = (d-1) * h/2+1$. Let $\troots$ be a time such that $rw_\pi(\troots)=1$ for all principal subtrees. \\
For $h\geq 3$, $\epsilon \in [-0.5,0.5]$, there exists a $fractional$ $pebbling$ $game$ {\bf sub-root pebbling} $\pi$ of $T_d^h$ such that the following conditions are true. 

(0) there exists a time $\troots$ such that $rw_\pi(\troots)=1$ for all subtrees

(1) $sw_\pi(t) \leq min_h-\epsilon$ for $t \leq \troots$

(2) $w.w_\pi(\troots) \leq min_h+\epsilon - d$

(3) Any white pebble weight at $\troots$ can be removed using $sw_\pi(t) \leq min_h+\epsilon$ for  $t > \troots$ 

(4) $b.rw_\pi(\troots)=1$ for at least one subtree

(5) $sw_\pi(t) \leq min_h+\epsilon$ for  $t > \troots$\\

\begin{obs} 
The previous IH resembles the IH for the lower bound to be proved later.
\end{obs}

The next two lemmas are to be used in the proof of the Induction hypothesis. They are to be applied to the subtrees of the root. They deal with leaving black pebble weight and removing white pebble weight.

\begin{lemma} \label{fup1}

It follows from the IH for height h, for $E \in (0,0.5]$, that there exists a {\bf pebbling} $\pi$ with $w_\pi(t) \leq min_h+E$ for all times $t$ and $w_\pi(0)=0$, that ends with $b.rw_{\pi} = 2E$ and $sw_\pi=0$.
\end{lemma}

\begin{lemma} \label{fup2}

It follows from the IH for height h, for $E \in (0,0.5]$, that there exists a {\bf pebbling} $\pi$ with $w_\pi(t) \leq min_h+E$ for all times $t$,  $w.rw_{\pi}(0) = 2E$  and $sw_\pi(0)=0$, that ends with $w_\pi=0$.
\end{lemma}

\noindent
 {\bf Proof of Lemma \ref{fup1}}\begin{proofL} 
 
We modify the pebbling in the IH with $\epsilon=-E$. We slide 2E black pebble weight to the root a step after $\troots$. This does not exceed $min_h+E$ weight since we use the same weight as at $\troots$, $w_\pi(\troots) \leq min_h +E$. 

We remove all black pebble weight and we use $sw_\pi \leq min_h-E$ to remove the remaining white pebble weight by condition 3 of IH. Thus for $t > \troots$, since we maintain $b.rw_{\pi}(t) = 2E$, we use $w_\pi(t) \leq min_h+E$. Thus we use $w_\pi(t) \leq min_h+E$ for all times t and have satisfied the conditions of the lemma.
 
 \end{proofL} 

\noindent
 {\bf Proof of Lemma \ref{fup2}}\begin{proofL} 
 
Given the white pebble weight on the root we follow the pebbling in the IH with $\epsilon=E$. We modify the pebbling by removing the pebble weight on the root at time $\troots$. We use $sw_\pi(t) \leq min_h-E$ for $t \leq \troots$ while maintaining $w.rw_{\pi}(t) = 2E$. 

We then remove all black pebble weight and use $sw_\pi(t) \leq min_h+E$ for $t > \troots$ to remove the white pebble weight by the IH. Thus we use $w_\pi(t) \leq min_h+E$ for all t.
 
 \end{proofL}

\noindent
{\bf Proof of the Induction Hypothesis}

\noindent
{\bf Base Case :} h = 3

In this case $min_h$ = $min_3$ = $3/2*(d-1) + 1$.

Let the nodes $v_i$ be the children of the root, $i \in [d]$. Let $v_{last}$ be the last node enumerated in this way.\\

For the first (d-1) $v_i$, place $(d-1)/2-\epsilon$ black pebble weight between them. This value is the amount in excess of d, the amount needed to pebble the leaves of the final subtree. Do this by placing d pebble weight on the leaves and sliding the largest possible portion of this amount to the subtree root (at most 1 per subtree root). Next, remove black pebble weight not on the subtree roots. Repeat starting with the first subtree until we place $(d-1)/2-\epsilon$ black pebble weight.

There are enough children of the root which are not $v_{last}$ to leave this amount since $(d-1)/2-\epsilon \leq (d-1)/2+(d-1)/2 = (d-1)$. 

We must use d pebble weight on the leaves each time we leave a fraction of a black pebble on a $v_i$. However, $(d-1)/2 -\epsilon + d = 3/2(d-1) -\epsilon + 1 = min_3 -\epsilon$. Thus we do not violate (1) in the IH when leaving $(d-1)/2-\epsilon$ black pebble weight on the first (d-1) $v_i$.

We then use d pebble weight on the leaves of $v_{last}$. We then slide one pebble weight to $v_{last}$ and remove the weight on the leaves.\\

We then add $(d-1)/2 + \epsilon$ white pebble weight to the first (d-1) $v_i$ to reach $\troots$.

At this time we have d pebble weight, thus we have not violated (1).

In this way $sw_\pi(t) \leq min_3-\epsilon$ for $t \leq \troots$ thus $\pi$ satisfies (1).

Since at this time $v_{last}$ is black pebbled (4) is satisfied.

Also $w.w_\pi(\troots) = (d-1)/2 + \epsilon = min_3 - d$, thus (2) is satisfied.

We then remove all black pebble weight.\\

We may then remove any of this white pebble weight using d pebble weight.

When we remove this white pebble weight we have $sw_\pi \leq (d-1)/2 + \epsilon + d = 3/2(d-1) + \epsilon + 1 = min_3 + \epsilon$ as required. Thus (3) is satisfied. 

Since this is all we must do and this is the most we use after $\troots$, condition (5) is satisfied.\\

Thus the specified $\pi$ satisfies all conditions and the IH is satisfied.\\

\noindent
{\bf Induction step :} We prove the induction hypothesis for $h+1$ assuming it for $h'$,
$3\le h' \le h$.\\
Note $min_{h+1} = min_h + (d-1)/2$.

Let $P_i$ be the the subtrees of the root, $i \in [d]$. Let $P_{last}$ be the last subtree enumerated in this way.\\ 

Using {\bf Lemma \ref{fup1}} we leave $(d-1)/2 - \epsilon$ pebble weight on the root of the first (d-1) subtrees.

If $(d-1)/2 - \epsilon \leq 1$. We leave $(d-1)/2 - \epsilon$ pebble weight on the last of the first (d-1) $P_i$. To do so we require $w_\pi \leq min_h +((d-1)/2 - \epsilon)/2$ by {\bf Lemma \ref{fup1}}. In the other subtrees we leave no pebble weight. Thus we do not exceed $min_h + (d-1)/2-\epsilon$ and do not violate (1).

If $(d-1)/2 - \epsilon > 1$. We leave one pebble weight on the root of the last of the first (d-1) $P_i$. Thus we require $min_h + 0.5$ by {\bf Lemma \ref{fup1}}. At this time we have $(d-1)/2 - \epsilon-1$ on the root of the other $P_i$. In the prior trees we require at most the same pebble weight while maintaining less in the other trees at that time. Thus we do not exceed $min_h + (d-1)/2-\epsilon$ and do not violate (1).\\

For the final subtree, we use the pebbling in the IH for height h, with $\epsilon = 0$, except we modify the pebbling to slide a pebble in the step after $\troots$. A slidable pebble exists by condition (4). We then remove all black pebbles in $P_{last}$ other than the black pebble on the root, leaving $min_h - d$ white pebble weight. Since we do not use more than pebble weight $min_h$ in $P_{last} $ while maintaining $(d-1)/2-\epsilon$ in the other subtrees, we do not violate (1).

We then use $(d-1)/2+\epsilon$ white pebble weight on the root of the other $P_i$ to reach $\troots$. At this time we have d pebble weight on the subtree roots while having $min_h - d$ white pebble weight in $P_{last}$. We thus have $min_h$ total pebble weight at this time and do not violate (1). \\

Thus, condition (1) is satisfied as we have $sw_\pi(t) \leq min_h-\epsilon$ for all $t \leq \troots$.

At this time we have $b.rw_\pi(\troots)[P_{last}] =1$, thus (4) is satisfied.\\

We then remove all black pebble weight.\\

We have $(d-1)/2+\epsilon$ white pebble weight on the roots of the subtrees while having $w.w_\pi(\troots)[P_{last}] = min_h-d$. Thus we have $w.w_\pi(\troots) \leq min_h +(d-1)/2+\epsilon -d = min_{h+1} + \epsilon - d$ and (2) is satisfied.

We first remove the white pebble weight from the subtree $P_{last}$. By IH, this requires $sw_\pi[P_{last}] \leq min_h$ while maintaining $(d-1)/2+\epsilon$ pebble weight in the other subtrees. Thus, to remove this white pebble weight we require $sw_\pi(t) \leq min_{h+1}+\epsilon$ for $t > \troots$.

We next remove white pebble weight from the first subtree with white pebble weight on the root, $P_{first}$. Suppose, $w.rw_\pi(\troots)[P_{first}] = 2E$. Using {\bf lemma \ref{fup2}} we can remove the white pebble weight using $w_\pi(t)[P_{first}] \leq min_h + E$. At this time we have less than $(d-1)/2+\epsilon - 2E$ pebble weight in the other trees. Thus $sw_\pi(t) \leq min_{h+1}+\epsilon$. We then remove the white pebble weight on the root of any remaining subtree in the same way.

Thus to remove the white pebble weight we required $sw_\pi(t) \leq min_{h+1}+\epsilon$ for $t>\troots$ and condition (3) is satisfied. Also, all times $t > \troots$, $sw_\pi(t) \leq min_{h+1}+\epsilon$ and (5) is satisfied.\\

Thus the specified pebbling $\pi$ satisfies all conditions and the IH is satisfied.\\\\

This result is obviously not possible without the use of fractional pebbles. Thus fractional pebbles allow for a large number of strategies that are not possible in other pebbling games. This gives us the intuition as to why we need a stronger induction hypothesis in the proof of the main lemma.\\\\

\subsection{Fractional Pebbling Game Lower Bounds} 

\noindent
We now prove the main theorem, which we state formally as :\\

\noindent
{\bf Main Theorem} \\Let $min_h = (d-1)h/2+1$.
For every {\bf root-pebbling} $\pi$ of $T^h_d$ there is a time $t$ such
that $w_\pi(t) \ge min_h$.\\

\noindent
The proof is simple for $h=2$. The proof for $h\geq3$ is by induction on $h$.\\

When Combined with the previous section we have a tight bound on the number of pebbles taken to complete the $fractional $ $pebbling$ $game$ for balanced $d$-ary trees. The result is new. Similar, but loose, lower bounds can be found in \cite{c:pebjournal}. In \cite{c:pebjournal}, they are the result of a reduction to a similar problem \cite{k:bwpyr}, we take a more direct approach.

The theorem is shown using the following induction hypothesis.\\

\noindent
{\bf Induction Hypothesis [IH$(h)$]:} Let $\pi$ be a {\bf sub-root sub-pebbling}
of $T^h_d$. Let $\troots$ be a time such that $rw_\pi$($\troots$)=1 for all principal subtrees.\\

If $h\geq 3$, $\epsilon \in (-0.5,0.5]$, $b.sw_{\pi}$(0) $\leq$ $1-\epsilon$, $b.rw_{\pi}$(0) = arbitrary, 
%(where $b.sw_{\pi}$(0) is the initial black subtree pebble weight) 
and $\pi$ is such that $sw_\pi(t) \leq min_h-\epsilon$ for $t \leq \troots$, then there is a time $\tbs > \troots$ such that $sw_\pi(\tbs) \geq min_h+\epsilon$ and $w.sw_\pi(t) \geq 0.5+\epsilon$ for t in $[\troots, \tbs]$.\\

\begin{tabular} { |l|l|l|}\hline
initial conditions & additional conditions & consequences\\\hline
$b.sw_{\pi}(0) \leq 1-\epsilon$ & $sw_\pi(t) \leq min_h-\epsilon$ for $t \leq \troots$ & $sw_\pi(\tbs) \geq min_h+\epsilon$\\\hline
$b.rw_{\pi}$(0) = arbitrary & & $w.sw_\pi(t) \geq 0.5+\epsilon$ for t in $[\troots, \tbs]$\\\hline
%$w.rw_{\pi}$(0) = arbitrary & & & & \\ \hline
%$w.sw_{\pi}$(0) = arbitrary & & & & \\ \hline
\end{tabular}\\\\

The Induction Hypothesis can be interpreted as indicating that we require more after if we use less before.

\noindent
\begin{obs}
The Induction Hypothesis implies the theorem. 
This is the case since we must at some time, $\troot$, have pebble weight 1 on the root in a {\bf root-pebbling}. If at $\troot$ the root has any black pebble weight we must have reached a time $\troots$ to place this black pebble weight. If it has only white pebble weight at $\troot$, we must reach a time $\troots$ to remove this white pebble weight. White pebble weight must be removed to satisfy the conditions of a {\bf root-pebbling}. It is therefore impossible to always use less than $min_h$ since by the Induction hypothesis we would need to use more than $min_h$ after $\troots$.\\
\end{obs}

\noindent
{\bf Proof of the Base Case of the Induction Hypothesis}  (h = 3)

In this case $min_h = min_3 = 3/2(d-1) + 1 = 3/2d -1/2$.

Let the nodes $v_i$ be the children of the root.\\

\noindent
{\bf Case I :} The black pebble weight on the $v_i$ is never increased at any time t such that t $\leq$
$\troots$.  

Then the total black pebble weight of the $v_i$ at
$\troots$ is at most $1-\epsilon$, so the white pebble weight for these nodes at
$\troots$ must be at least $d-(1-\epsilon) = d-1+ \epsilon$.

Let $\tbs$ be the first time we remove white pebble weight after $\troots$. Since we must have pebble weight 1 on all of the children to remove white pebble weight we have that the total pebble weight required to remove white
pebble weight is at least
$d+(d - 1 + \epsilon) = 2d - 1 + \epsilon > 3/2d - 1/2 + \epsilon = min_h + \epsilon$ at time $\tbs$. 

$\tbs > \troots$, since at $\troots$ the pebble weight on the $v_i$ is d, thus at this time we could not have had the required pebble weight on the children due to the restriction on total pebble weight.

Also, during the interval [$\troots$, $\tbs$],
$w.sw_\pi(t) \geq (d-1) + \epsilon > 0.5 + \epsilon$, as required. 

Thus the IH is satisfied in this case.\\

\noindent
{\bf Case II :}  The black pebble weight on the nodes $v_i$ is increased at some time t such that $t \leq
\troots$.  

Let $t_a$* be one step before the last time of such an increase.
Let $\alpha$ be the total black pebble weight of the $v_i$ at time $t_a$*.
Then the total subtree pebble weight at time $t_a$* is at least d+$\alpha$,
which by assumption is at most $min_h - \epsilon$.  Therefore, d+ $\alpha$ $\le$ 3/2d -1/2 - $\epsilon$, and hence 
\begin{equation}\alpha \le 1/2d - 1/2 - \epsilon \label{05eps}\end{equation}
      
After this increase at time $t_a$* the total black pebble weight of the $v_i$
is at most 1 + $\alpha$.  Hence the white pebble weight of the $v_i$ at
$\troots$ satisfies $w.sw_\pi$($\troots$) $\ge$ d-(1 + $\alpha$) = d-1-$\alpha$.
    
Let $\tbs$ be the time just before the first time after $\troots$ that this
white pebble weight is decreased.  Since we need d pebble weight on the leaves at such a time,\\
$sw_\pi$($\tbs$) $\ge$ d+(d-1-$\alpha$) \\
= 2d -1-$\alpha$ \\
$\ge$ 2d -1 - 1/2d + 1/2 + $\epsilon$ (by $\ref{05eps}$)\\
= 3/2d - 1/2 + $\epsilon$ \\
= $min_h + \epsilon$, as required.

Also, $\tbs$ $>$ $\troots$, since at $\troots$ the pebble weight on the $v_i$ is d, thus we could not have had the required pebble weight on the children due to the restriction on total pebble weight.

Finally, during the interval [$\troots$, $\tbs$], $w.sw_\pi(t) \geq d-1-\alpha$ $\ge$ $d-1 - (1/2d - 1/2 - \epsilon)$  = $1/2d -1/2  + \epsilon$ $\geq$ $0.5 +  \epsilon$, as required $(d \geq 2)$. Thus the IH is satisfied in this case.\\

Thus, in the base case the IH is satisfied.\\\\

The next two lemmas are to be used in the proof of the induction step. They are to be applied to the subtrees of the root.

\begin{lemma} \label{flb1}
Let $\pi$ be a {\bf root sub-pebbling}
of $T^h_d$. Let $\troot$ be any time such that $rw_\pi$($\troot$) = 1.

It follows from the IH for height h, that if $E \in [0.0, 0.5)$, $b.sw_{\pi}$(0) $\leq$ $0.5+E$, $b.rw_{\pi}$(0) $\leq$ $2E$ and $\pi$ is such that $sw_\pi(t) \leq min_h-0.5+E$ for t $\leq$ $\troot$, then there is a time $\tbss$, such that $\troot<\tbss$, $w_\pi(\tbss) \geq min_h+0.5-E$ and $w.w_\pi(t) \geq 1-2E$ for t in $[\troot, \tbss]$.\\
%{\bf Lemma 1} It follows from IH for height h, that if $(0.5-E) \in (0.0, 0.5]$, $sw_{\pi}$(0) $\leq$ $1-(0.5-E)$, $b.rw_{\pi}$(0) $\leq$ $1-2(0.5-E)$ and $w_\pi(t) \leq min_h-(0.5-E)$ for t before $\troot$, then at some time $\tbss>\troot$, $w_\pi(\tbss) \geq min_h+(0.5-E)$ and $w_\pi(t) \geq 2(0.5-E)$ for t in $[\troot, \tbss]$.\\

\end{lemma}

\begin{tabular} { |l|l|l|}\hline
initial conditions & additional conditions & consequences\\ \hline
$b.sw_{\pi}(0) \leq 0.5+E$ & $sw_\pi(t) \leq min_h-0.5+E$ for t $\leq$ $\troot$ & $w_\pi(\tbss) \geq min_h+0.5-E$\\ \hline
$b.rw_{\pi}$(0) $\leq$ $2E$ & &  $w.w_\pi(t) \geq 1-2E$ for t in $[\troot, \tbss]$\\ \hline
%$w.rw_{\pi}$(0) = arbitrary & & & & \\ \hline
%$w.sw_{\pi}$(0) = arbitrary & & & & \\ \hline
\end{tabular}\\\\

\begin{lemma} \label{flb2}
Let $\pi$ be a {\bf root sub-pebbling}
of $T^h_d$. Let $\troot$ be any time such that $rw_\pi$($\troot$) = 1.

It follows from the IH for height h, that if $E \in [0, 1)$, $b.sw_{\pi}$(0) $\leq$ $0.5 + E$, at some time $t_0$, 0 $\leq$ $t_0$ $\leq$ $\troot$, $b.rw_{\pi}$($t_0$) $\leq$ $E$ and $\pi$ is such that $w_\pi(t) \leq min_h-0.5+E$ for t $\leq$ $\troot$, then there is a time $\tbss$, such that $\troot<\tbss$, $w_\pi(\tbss) \geq min_h+0.5-E$ and $w.w_\pi(t) \geq 1-E$ for t in $[\troot, \tbss]$.\\
%{\bf Lemma 1} It follows from IH for height h, that if $\epsilon \in (0.0, 0.5]$, $sw_{\pi}$(0) $\leq$ $1-\epsilon$, $b.rw_{\pi}$(0) $\leq$ $1-2\epsilon$ and $w_\pi(t) \leq min_h-\epsilon$ for t before $\troot$, then at some time $\tbss>\troot$, $w_\pi(\tbss) \geq min_h+\epsilon$ and $w_\pi(t) \geq 2\epsilon$ for t in $[\troot, \tbss]$.\\

\end{lemma}

\begin{tabular} { |l|l|l|}\hline
initial conditions & additional conditions & consequences\\ \hline
$b.sw_{\pi}$(0) $\leq$ $0.5 + E$ & $w_\pi(t) \leq min_h-0.5+E$ for t $\leq$ $\troot$ & $w_\pi(\tbss) \geq min_h+0.5-E$\\ \hline
$b.rw_{\pi}$($t_0$) $\leq$ $E$, $t_0$ $\leq$ $\troot$ & & $w.w_\pi(t) \geq 1-E$ for t in $[\troot, \tbss]$\\ \hline
%$w.rw_{\pi}$(0) = arbitrary & & & & \\ \hline
%$w.sw_{\pi}$(0) = arbitrary & & & & \\ \hline
\end{tabular}\\\\

\noindent
We make the following observations : 

\begin{obs}
In {\bf Lemma \ref{flb1}} additional initial black pebble weight on the root allows us to use less white pebble weight for t in [$\troot$, $\tbss$] than in {\bf Lemma \ref{flb2}}.
\end{obs}

\begin{obs}
In {\bf Lemma \ref{flb2}} we introduce a time $t_0$. There may be more black pebble weight on the root before time $t_0$, however, it can not help us achieve the specified $\troot$ if it is removed before $\troot$.\\
\end{obs}

\begin{obs}
The IH implies conditions on the subtree pebble weight while the lemmas imply conditions on pebble weight anywhere.\\
\end{obs}

\begin{obs}
The IH allows for arbitrary black root weight. Given the allowed pebbling moves, black root weight can not help us achieve $\troots$. This is not the case in the lemmas, it is possible that black root weight helps us attain $\troot$.\\
\end{obs}

\noindent
 {\bf Proof of Lemma \ref{flb1}}\begin{proofL} 

 {\bf Lemma \ref{flb1}} will be used in the induction step since it is possible to leave some pebble weight on one subtree and proceed with the pebbling in the other subtrees.\\

%\noindent
%{\bf Case 1} We removed the initial pebble weight we do as well as we did in the Lemma 0. In this case, choosing $\tbss$ to be the $t_b$ in lemma 0, we would have $w_\pi(t) \geq 0.5+(0.5-E)$ for t in [$\troot$, $\tbss$] and $0.5+(0.5-E) > 2(0.5-E)$. Thus in this case the lemma is satisfied.\\

%\noindent
%{\bf Case 2} If we remove a portion of the pebble weight or started with less on the root and only leave b.rw=E then we can $sw_\pi(t) \leq min_h-(0.5-E)-E$ for times before $\troot$. If we reach a $\troots$ before $\troot$ we will need $sw_\pi(t_b) \geq$ min($min_h+0.5$, $min_h+(0.5-E)+E$) at a time $t_b$ after $\troot$ by IH and have min(1, $0.5+(0.5-E)+E$) until this time. This will exceed the minimum specified in the Lemma. Also min(1, $0.5+(0.5-E)+E$) $\geq$ 2$(0.5-E)$ for all possible $(0.5-E)$ and E, thus lemma 1 is satisfied.

%If we instead add a white pebble of weight 1-E to reach $\troot$ we would be better off adding less pebble weight to the root and maintaining all of the initial pebble weight since we would be able to use more pebble weight in the subtrees.\\

%\noindent
%By the assumption of a pebbling this pebble weight is placed to the amount so that it may be completed by a white pebble. Thus we may not remove it. Another $\troots$ must occur since we must add black pebble weight to reach $\troot$ or remove white pebble weight after reaching $\troot$.\\

We must reach a time $\troots$, either to add black pebble weight to reach $\troot$ or to remove white pebble weight added to reach $\troot$. Since times $\troots$ exist, $\pi$ is also a {\bf sub-root sub-pebbling}. Thus we will apply the IH at these points denoted $\troots$.\\

\noindent
{\bf Case 1 :} $\exists \troots, \troots \leq \troot$. 

By IH with $\epsilon=0.5-E$, since by assumption $sw_\pi(t) \leq min_h-0.5+E$ for $t \leq \troot$ and $b.sw_{\pi}(0) \leq 0.5+E$, then at some time $\tbss=\tbs$, $sw_\pi(\tbss ) \geq min_h+0.5-E$ and $w.w_\pi(t) \geq  1-E$ for t in [$\troots$, $\tbss$]. Also, $1-E \geq 1-2E$ since $E \geq 0$. 

Since $min_h+0.5-E > min_h-0.5+E$ for all allowed E, we have not been allotted enough pebbles before $\troot$ and $\troot < \tbss$ . 

Thus the conditions of the lemma are satisfied.\\

\noindent
{\bf Case 2 :} $\forall \troots, \troot < \troots$. Then, to reach $\troot$ we must use white pebble weight. Since $b.rw_{\pi}(0) \leq 2E$, $w.rw_\pi(\troot) \geq 1-2E$. We must then reach a $\troots$ to remove this white pebble weight. Let $\trootsFirst$ be the first such $\troots$. Thus,
\begin{equation}w.rw_\pi(t) \geq 1-2E\textrm{ for t in }[\troot, \trootsFirst] \label{case2l1whiteC}\end{equation}

\noindent
{\bf Case 2-A :} $\exists t, t \in (\troot, \trootsFirst]$ and $sw_\pi(t) \geq min_h - 0.5+E$

Choose $\tbss$ to be the first such $t$. Then $w_\pi(\tbss) \geq min_h + 0.5-E$ and $w.w_\pi(t) \geq 1-2E$ for times t in [$\troot$, $\tbss$] since we have yet to remove the white pebble weight on the root (\ref{case2l1whiteC}). Thus the lemma is satisfied in this case.\\

\noindent
{\bf Case 2-B :} $\forall t, if~t~\in (\troot, \trootsFirst]$ then $sw_\pi(t) < min_h - 0.5+E$

Then $sw_\pi(t) \leq min_h-0.5+E$ for t in [0, $\trootsFirst$]. By IH with $\epsilon=0.5-E$, we have some time $\tbs$ $>$ $\trootsFirst$ such that $sw_\pi(\tbs) \geq min_h + 0.5-E$ and $w.w_\pi(t) \geq 1-E$ for t in [$\trootsFirst$,  $\tbs$]. We choose $\tbss = \tbs$. 

$w.w_\pi(t) \geq 1-2E$ for times t in [$\troot$, $\trootsFirst$] (\ref{case2l1whiteC}). Thus, $w.w_\pi(t) \geq 1-2E$ for t in $[\troot, \tbss]$. Thus, all conditions are met and the lemma is satisfied in this case.\\

Thus {\bf Lemma \ref{flb1}} is satisfied in all cases.

\end{proofL}\\

\noindent
 {\bf Proof of Lemma \ref{flb2}}\begin{proofL} 
 
Lemma \ref{flb2} is to be used in the induction step when we increase the pebble weight on the root of the subtrees.\\

We must reach a time $\troots$, either to add black pebble weight to reach $\troot$ or to remove white pebble weight added to reach $\troot$. Since these times exist, $\pi$ is also a {\bf sub-root sub-pebbling}. Thus we will apply the IH at these times denoted $\troots$.\\

\noindent
{\bf Case 1 :} $\troots$ $\leq$ $\troot$ $<$ $\tbs$ for some $\troots$ and corresponding $\tbs$.

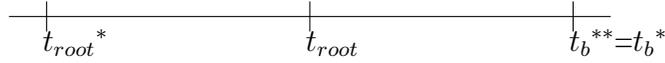
\begin{figure}[H]
  \centering

\begin{tikzpicture}

\draw [-] (0,0) -- (14,0)[color=white] node [below] {};
\draw [-] (4,0) -- (12,0) node [below] {};
\foreach \x in {4.5, 8, 11.5} 
\draw (\x, 0.2) -- (\x,-0.2) node[anchor=north] {}; 
    \node[text width=4 cm,text ragged, text height = 5, anchor=west]  at (4.3,-0.4) {$\troots$};
    \node[text width=4 cm,text ragged, text height = 5, anchor=west]  at (7.8,-0.4) {$\troot$};
    \node[text width=4 cm,text ragged, text height = 5, anchor=west]  at (11.3,-0.4) {$\tbss$=$\tbs$};
\end{tikzpicture}

\caption{Timeline for Case 1. In this case we reach $\troots$ before $\troot$ and do not reach the corresponding $\tbs$ until after $\troot$}
\end{figure}

By IH, taking $\epsilon$ to be $0.5-E$, taking $\tbss$ = $\tbs$, since $sw_{\pi}$(t) $\leq$ $min_h -0.5 + E$ for $t \leq \troots$ and $b.sw_{\pi}(0) \leq 0.5 + E$, then $sw_{\pi}$($\tbss$) $\geq$ $min_h + 0.5-E$ and $w.w_\pi(t) \geq 1-E$ for t in $[\troots, \tbss]$. By assumption we also have $\troot$ $<$ $\tbss$. Thus in this case the lemma is satisfied.\\

\noindent
{\bf Case 2 :} $\forall$ $\troots$, $\troot$ $<$ $\troots$.

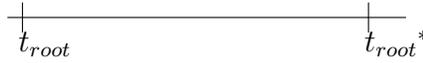
\begin{figure}[H]
  \centering
%\%includegraphics{./thesisimgs/lemma0caseI2set.png}
\begin{tikzpicture}

\draw [-] (0,0) -- (14,0)[color=white] node [below] {};
\draw [-] (4,0) -- (9.3,0) node [below] {};
\foreach \x in {4.2,8.8} 
\draw (\x, 0.2) -- (\x,-0.2) node[anchor=north] {}; 
    \node[text width=4 cm,text ragged, text height = 5, anchor=west]  at (4,-0.4) {$\troot$};
    \node[text width=4 cm,text ragged, text height = 5, anchor=west]  at (8.6,-0.4) {$\troots$};
\end{tikzpicture}

\caption{Setup for Case 2.}
\end{figure}

Then we use white pebble weight to reach $\troot$,
\begin{equation}w.rw_{\pi}(t)=1-E\textrm{ for t in }[\troot, \trootsFirst]  \label{case3whiteC}\end{equation}

Let $\trootsFirst$ be the first $\troots$. \\

\noindent
{\bf Case 2-A :} $\exists$ t, t $\in$ ($\troot$, $\trootsFirst$] and $sw_{\pi}$(t) $\geq$ $min_h-0.5$ 

We let $\tbss$ be such a time t. Then we meet the criteria in the lemma since we have $w_\pi(\tbss) \geq min_h+0.5-E$ and $w.w_\pi(t) \geq 1-E$ for t in $[\troot, \tbss]$ (\ref{case3whiteC}). Thus the lemma is satisfied in this case.\\

\noindent
{\bf Case 2-B :} $\forall$ t, if t $\in$ ($\troot$, $\trootsFirst$] then $sw_{\pi}$(t) $<$ $min_h-0.5$

$min_h-0.5$ $\leq$ $min_h-0.5 + E$ for all allowed $E$. We have used $sw_{\pi}$(t) $\leq$ $min_h-0.5 + E$ for t in [0, $\trootsFirst$]. By the IH, taking $\epsilon$ to be $0.5-E$,  letting $\tbss$ = $\tbs$, we must use $sw_{\pi}$($\tbss$) $\geq$ $min_h+0.5-E$ at $\tbss$ $>$ $\trootsFirst$. 

Also by the IH $w.w_{\pi}$(t) $\geq$ $1-E$ for t in [$\trootsFirst$, $\tbss $]. $w.rw_{\pi}$(t) $\geq$ 1-E for t in [$\troot$, $\trootsFirst$] (\ref{case3whiteC}), thus $w.w_{\pi}$(t) $\geq$ $1-E$ for t in [$\troot$, $\tbss$]. Thus the lemma is satisfied in this case.\\

\begin{figure}[H]
  \centering
%\%includegraphics{./thesisimgs/lemma0caseI2.png}
\begin{tikzpicture}

\draw [-] (0,0) -- (12,0)[color=white] node [below] {};
\draw [-] (4,0) -- (14,0) node [below] {};
\foreach \x in {6.5,10, 13.5} 
\draw (\x, 0.2) -- (\x,-0.2) node[anchor=north] {}; 
    \node[text width=4 cm,text ragged, text height = 5, anchor=west]  at (6.3,-0.4) {$\troot$};
    \node[text width=4 cm,text ragged, text height = 5, anchor=west]  at (9.8,-0.4) {$\troots$};
    \node[text width=4 cm,text ragged, text height = 5, anchor=west]  at (13.3,-0.4) {$\tbs$ = $\tbss$};
    
\draw [-] (4,-1) -- (10,-1) node [below] {};
\foreach \x in {4,10} 
\draw (\x, -0.8) -- (\x,-1.2) node[anchor=north] {}; 
    \node[text width=4 cm,text ragged, text height = 5, anchor=west]  at (6,-1.1) {no $\troots$};
    
\draw [-] (6.5,-2) -- (10,-2) node [below] {};
\foreach \x in {6.5, 10} 
\draw (\x, -1.8) -- (\x,-2.2) node[anchor=north] {}; 
    \node[text width=4 cm,text ragged, text height = 5, anchor=west]  at (6.5,-2.2) {$w.rw_{\pi}$(t) $\geq$ $1-E$};
\end{tikzpicture}

\caption{Timeline for Case 2. As mentioned, 1-E pebble weight is on the root between $\troot$ and $\troots$.}
\end{figure}
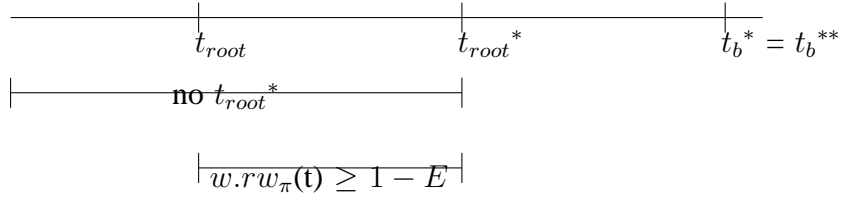

\noindent
{\bf Case 3 :} $\troots$ $<$ $\tbs$ $\leq$ $\troot$ for the last $\troots$ and corresponding $\tbs$ before $\troot$.\\

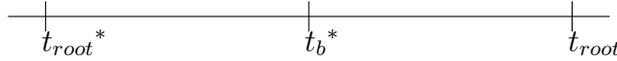
\begin{figure}[H]
  \centering
%\%includegraphics{./thesisimgs/lemma0caseII2set.png}
\begin{tikzpicture}

\draw [-] (0,0) -- (12,0)[color=white] node [below] {};
\draw [-] (4,0) -- (12,0) node [below] {};
\foreach \x in {4.5,8, 11.5} 
\draw (\x, 0.2) -- (\x,-0.2) node[anchor=north] {}; 
    \node[text width=4 cm,text ragged, text height = 5, anchor=west]  at (4.3,-0.4) {$\troots$};
    \node[text width=4 cm,text ragged, text height = 5, anchor=west]  at (7.8,-0.4) {$\tbs$};
    \node[text width=4 cm,text ragged, text height = 5, anchor=west]  at (11.3,-0.4) {$\troot$};
\end{tikzpicture}
\caption{Setup for Case 3.}
\end{figure}

\noindent
{\bf Case 3-A :} $E < 0.5$. By IH, taking $\epsilon$ to be $0.5-E$, since $sw_{\pi}$(t) $\leq$ $min_h -0.5 + E$ for $t \leq \troots$ and $b.sw_{\pi}(0) \leq 0.5 + E$, then $sw_{\pi}$($\tbs$) $\geq$ $min_h + 0.5-E$. However, $min_h+0.5-E > min_h-0.5 + E$. Thus we have not been allotted enough pebble weight before $\troot$ and we must proceed past $\troot$ before we may reach $\tbs$. Thus when $0.5 > E$, {\bf Case 3} is not possible. \\

\noindent
{\bf Case 3-B :} $E \geq 0.5$.

By IH, taking $\epsilon$ to be $0.5-E$, we must have a $\tbs$ such that $sw_{\pi}$($\tbs$) $\geq$ $min_h + 0.5-E$. 

At this time, $b.rw_{\pi}$($\tbs$) $\leq$ $2E-1$ $<$ 1 due to the restriction on total pebble weight before $\troot$. Since the chosen $\troots$ was the last before $\troot$ we must use white pebble weight to reach $\troot$, $w.rw_{\pi}(\troot) \geq 2-2E$.

Since this is not 0 we will need to reach another $\troots$ after $\troot$ to remove this white pebble weight. Since 2-2E $\geq$ 1-E, this case follows by the same argument in {\bf Case 2-A} and {\bf Case 2-B} .\\

Thus in all cases Lemma \ref{flb2} follows from IH.

\end{proofL}\\

\noindent
{\bf Induction step :} We prove the induction hypothesis for h+1 assuming it for h',
$3\le h' \le h$.

Fix $\pi =$ 0, ..., $\troots$, ... to be a {\bf sub-root sub-pebbling} of $T^{h+1}_d$ with $\troots$ such that $rw_\pi$($\troots$)=1 for all principal subtrees, and with
 \begin{equation}sw_{\pi}(t) \leq min_{h+1} - \epsilon = (d-1)(h+1)/2 + 1- \epsilon=  min_h +  (d-1)/2 - \epsilon~for~t~in~[0, \troots]\label{rtotal}\end{equation} 

\noindent
Further, we assume,

\begin{equation}\epsilon \in (-0.5,0.5]\label{aeps}\end{equation} 
\begin{equation}b.sw_{\pi}(0) \leq 1-\epsilon\label{swinit}\end{equation}

Let $P_i$ be the principal subtrees of $T^{h+1}_d$.
The restriction of $\pi$ to each of these subtrees is a valid
pebbling of that subtree.\\

\noindent
{\bf Case 1 :} $\forall t$, $\forall i$, if $t \leq \troots$ then $sw_{\pi}$(t)[$P_i$] $<$ $min_h-0.5$ 

For each principal subtree we will apply {\bf Lemma \ref{flb1}}. We will show that if we consider all subtrees this implies the desired bounds.

In this case, the subtree pebble weight of all subtrees  $P_i$ is less than $min_h-0.5$.

We have at most $1 - \epsilon$ initial black pebble weight in the $P_i$ by assumption ($\ref{swinit}$). We will separate this pebble weight between the subtrees and apply {\bf Lemma \ref{flb1}} to each subtree. Let us have $b.w_{\pi}$(0)[$P_i$] = 2$E_i$. We choose to express the amount this way since it resemble amounts expressed in {\bf Lemma \ref{flb1}}.\\

It is the case that $E_i$ $\geq$ 0 since pebble weight is non-negative.

If 0 ${\leq}$ $E_i$ ${<}$ 0.5 we may apply {\bf Lemma \ref{flb1}} to the $i^{th}$ subtree. Let G be the set of all $i$ such that 0 ${\leq}$ $E_i$ ${<}$ 0.5. We have $\Sigma_{i \in G} 1 - 2E_i \geq \Sigma_{i=1}^d 1- 2E_i$ since $0 \geq 1 - 2E_i$ for i $\notin$ G. 

The way in which we will use G will affirm that maintaining more than 1 black pebble weight in any tree is useless.

Note, G is not the empty set since $b.sw_{\pi}(0) \leq 1-\epsilon$ and d $\geq$ 2.\\

\noindent 
Note,\\
$\Sigma_{i=1}^d 2E_i \leq 1-\epsilon$, by construction,\\
$- \Sigma_{i=1}^d 2E_i  \geq - 1 + \epsilon$, then,\\
$\Sigma_{i=1}^d 1 - 2E_i  \geq d - 1 + \epsilon$, then,\\
\begin{equation}\Sigma_{i \in G} (1 - 2E_i) \geq d - 1 + \epsilon \label{f1case0IHWC}\end{equation}

For each subtree, we take $\troot$ in the lemma to be the time $\troots$. This is possible since $rw_\pi$($\troots$)[$P_i$]=1 as required by {\bf Lemma \ref{flb1}}.\\

We apply {\bf Lemma \ref{flb1}} to $P_i$, i $\in$ G,  taking E in the lemma to be ${E_i}$ and with $\tbith$ := $\tbss$. Then, $\tbith $ $>$ $\troots$, $w_{\pi}$($\tbith$)[$P_i$] $\geq$ $min_h+0.5-E_i$ and $w.w_\pi(t)$[$P_i$] $\geq 1-2E_i$ for t in [$\troots$, $\tbith$]. \\

We let $\tbs$=min($\tbith$) for $i$ $\in$ G.

We define $first$ to be this $i$. It is the first $\tbith$ we reach in $\pi$. Then we require $min_h+0.5-E_{first}$ in $P_{first}$ while maintaining at least 1-2$E_i$ in the remaining $P_i$, $i$ $\in$ G and $i$ $\ne$ $first$. Then,\\
$sw_{\pi}$($\tbs$) $\geq$ $min_h+0.5-E_{first}+\Sigma_{i \in G, i \ne first} (1-2E_i)$\\
$\geq min_h+0.5-2E_{first}+\Sigma_{i \in G, i \ne first} (1-2E_i)$  (since, $0 \geq -E_{first}$)\\
= $min_h-0.5+\Sigma_{i \in G} (1-2E_i)$\\
$\geq min_h-0.5+d-1+\epsilon$ (by \ref{f1case0IHWC})\\
$\geq$ $min_h-(d-1)/2+(d-1)+\epsilon$ (since $d \geq 2$)\\
= $min_h+(d-1)/2+\epsilon$\\
= $min_{h+1}+\epsilon$

Additionally, \\$w.sw_\pi(t) \geq$ $\Sigma_{i \in G} (1 - 2E_i) \geq d - 1 + \epsilon \geq 1 + \epsilon$ for t in $[\troots, \tbs]$ (by \ref{f1case0IHWC}). 

Thus the IH is satisfied in {\bf Case 1}.\\\\

\noindent
{\bf Case 2 :}  $\exists t$, $\exists i$, $t \le \troots$ and $sw_{\pi}(t)[P_i] \geq min_h-0.5$.

For each principal subtree we will try to apply one of the lemmas. We will then show that taken together this results in the desired bounds. Also recall that we fixed $\pi$ = 0, ..., $\troots$, ... .\\

Suppose $sw_{\pi}$(t) $\geq$ $min_h-0.5$ for the last time before $\troots$ in the subtree $P_{last}$. Let this time be $t_{last}$. Then $t_{last}$ $\leq$ $\troots$ and 
\begin{equation} sw_{\pi}(t_{last})[P_{last}] \geq min_h-0.5\label{swp2tlast}\end{equation}

For any value $r_i$, for all $i$ $\ne$ $last$, define $t_{r_i}$ to be the last time in [0, $\troots$] such that  $sw_{\pi}$($t_{r_i}$)[$P_i$]  $\geq$ $min_h-0.5+r_i$ or the initial time if no such time exists.

Define $R_i$ to be the max $r_i$ such that $w_{\pi}$(t)[$P_i$] $\geq$ 2${r_i}$ for times t in [$t_{r_i}$, $\troots$].\\

There is always a time $\troots$ since $\pi$ is a {\bf sub-root sub-pebbling}. The described condition is true for some value of $r_i$ as it is true for $r_i$ = 0 and this is the smallest value possible. There is therefore always a time $t_{R_i}$ for each principal subtree. Thus,
\begin{equation} {R_i} \geq 0\label{DboundsC1A}\end{equation}

By definition of $t_{R_i}$ and $t_{last}$, 
\begin{equation} t_{R_i} < t_{last}\label{tdtlast}\end{equation}

%This is since $t_{last}$ is the last time in $\pi$ we use the amount described at $t_{R_i}$. 

This is a result of the restriction on total pebble weight (\ref{rtotal}) and having at least $min_h-0.5$ pebble weight in $P_{last}$ at $t_{last}$. We show that we must have less pebble weight than $min_h-0.5$ in the other subtrees at $t_{last}$. Suppose we did not, we then have at least $2min_h -1$ total pebble weight.\\

\noindent
$sw_{\pi}(t) \geq 2min_h -1$\\
$= min_h +(d-1)h/2 + 1 - 1$\\
$= min_h +(d-1)h/2$\\
$> min_h +(d-1)$   (Since $h > 2$)\\
$\geq  min_h +(d-1)/2 - \epsilon$\\
$= min_{h+1}- \epsilon$

This would contradict the assumption for total subtree pebble weight (\ref{rtotal}). Thus $t_{last}$ is the last time in $\pi$ we use the amount described at $t_{R_i}$ and (\ref{tdtlast}) holds.\\

In summary, the choice of ${R_i}$ implies the following, 

\begin{equation}sw_{\pi}(t_{R_i})[P_i]  \geq min_h-0.5+{R_i}~or~t_{R_i} = 0\label{tdamt}\end{equation}
\begin{equation}w_{\pi}(t)[P_i] \geq 2{R_i}~for~t~in~[t_{R_i}, \troots]\label{tdres}\end{equation}

\begin{definition}\label{ddr2}

For each $i \ne last$, define $t_{Pi-init}$ to be a time such that $w_{\pi}$($t_{Pi-init}$)[$P_i$] $\leq$ 2$R_i$ and $sw_{\pi}$(t)[$P_i$] $\leq$ $min_h-0.5+R_i$ fot t in [$t_{Pi-init}$, $\troots$]. 
\end{definition}

This will be useful since we wish to apply {\bf Lemma \ref{flb1}} to $P_i$ later with E = $R_i$ and initial time $t_{Pi-init}$. We show such a time always exists.

\noindent
{\bf Case I} : $w_{\pi}$(t)[$P_i$] = 2$R_i$ for some t in [$t_{R_i}$, $\troots$]. We let this time be $t_{Pi-init}$.

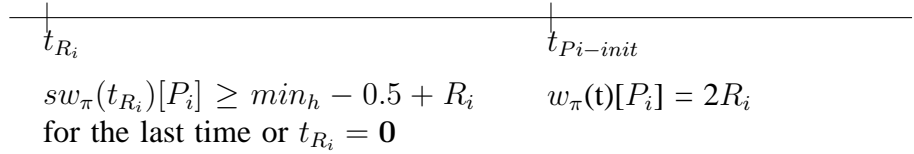
\begin{figure}[H]
  \centering
%\%includegraphics{./thesisimgs/case1B1i.png}
\begin{tikzpicture}
\draw [-] (0,0) -- (12,0) node [below] {};
\foreach \x in {0.5,7.2} 
\draw (\x, 0.2) -- (\x,-0.2) node[anchor=north] {};

    \node[text width=6 cm,text ragged, text height = 5, anchor=west]  at (0.3,-0.4) {$t_{R_i}$};
    \node[text width=6 cm,text ragged, text height = 5, anchor=west]  at (7,-0.4) {$t_{Pi-init}$};
    \node[text width=6 cm, text ragged, text height = 5, text depth = 25, anchor=west] at (0.3,-1.5){$sw_{\pi}(t_{R_i})[P_i]  \geq min_h-0.5+{R_i}$ for the last time or $t_{R_i} = {\bf 0}$}; 
    \node[text width=6 cm, text ragged,  text height = 6, text depth = 25, anchor=west] at (7,-1.5) {$w_{\pi}$(t)[$P_i$] = 2$R_i$};
    
\end{tikzpicture}
\caption{Depicts the situation in $P_i$ for {\bf Case I}.}
\end{figure}

\noindent
{\bf Case II} :  $w_{\pi}$(t)[$P_i$] $>$ 2$R_i$ for all times t in [$t_{R_i}$, $\troots$].

Then $sw_{\pi}$($t_{R_i}$)[$P_i$] = $min_h-0.5+R_i$. If this was not the case, the conditions would be true for a greater value of ${R_i}$ and we would have a contradiction. For similar reasons, $t_{R_i}$ is not the initial time else the condition would be true for a larger value of $R_i$.

Let $t_{before-R_i}$ be the last time such that $sw_{\pi}$($t_{before-R_i}$)[$P_i$] $>$ $min_h-0.5+R_i$ or the initial time if no such time exists. Then $t_{before-R_i}$ $<$ $t_{R_i}$. There must have been a time, $t_{Pi-init}$,  in [$t_{before-R_i}$, $t_{R_i}$] such that $w_{\pi}$($t_{Pi-init}$)[$P_i$] $\leq$ 2$R_i$. If this were not the case, the conditions would be true for a greater value of ${R_i}$ since we would have $w_{\pi}$(t)[$P_i$] $>$ 2${R_i}$ for t in [$t_{before-R_i}$, $\troots$] using the assumption in {\bf Case II}. Thus, the chosen $t_{Pi-init}$ satisfies the necessary conditions.

%\%includegraphics{./thesisimgs/case1B2i.png}

\begin{figure}[H]
  \centering
\begin{tikzpicture}
\draw [-] (0,0) -- (13,0) node [below] {};
\foreach \x in {0.5,5.2, 9.9} 
\draw (\x, 0.2) -- (\x,-0.2) node[anchor=north] {};

    \node[text width=4 cm,text ragged, text height = 5, anchor=west]  at (0.3,-0.4) {$t_{before-R_i}$};
    \node[text width=4 cm,text ragged, text height = 5, anchor=west]  at (5,-0.4) {$t_{Pi-init}$};
    \node[text width=4 cm,text ragged, text height = 5, anchor=west]  at (9.7,-0.4) {$t_{R_i}$};
    \node[text width=4.7 cm, text ragged, text height = 5, text depth = 25, anchor=west] at (0.3,-1.5) {$sw_{\pi}$($t_{before-R_i}$)[$P_i$] $>$ $min_h-0.5+R_i$ for the last time or $t_{before-R_i} = {\bf 0}$}; 
    \node[text width=4 cm, text ragged,  text height = 5, text depth = 25, anchor=west] at (5,-1.5) {$w_{\pi}$($t_{Pi-init}$)[$P_i$] $\leq$ 2$R_i$};
    \node[text width=4 cm, text ragged,  text height = 5, text depth = 25, anchor=west] at (9.7,-1.5) {$sw_{\pi}$($t_{R_i}$)[$P_i$] = $min_h-0.5+R_i$};
    
\end{tikzpicture}
\caption{Depicts the situation in $P_i$ for {\bf Case II}.}
\end{figure}
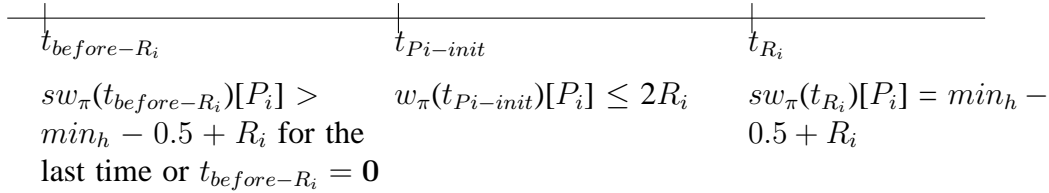

Thus in all cases, such a $t_{Pi-init}$ exists.\\\\

Let G be the set of all $i$ such that 0 ${\leq}$ $R_i$ ${<}$ 0.5, $i$ $\ne$ $last$. Since $2R_i \geq 1$ for $i \notin G$,
\begin{equation} \Sigma_{i=1, i \ne last}^d 2R_i \geq  (d-1-|G|) + \Sigma_{i \in G} 2R_i\label{GtonotG}\end{equation}
We will apply {\bf Lemma \ref{flb1}} to $P_i$ for $i$ $\in$ G, taking the initial time in the lemma to be $t_{Pi-init}$ and taking E in the lemma to be ${R_i}$.

We use $sw_{\pi}(t_{last})[P_{last}] \geq min_h-0.5$ (\ref{swp2tlast}) while maintaining $\Sigma_{i=1, i \ne last}^d 2R_i$ in the other $P_i$ at time $t_{last}$ ($\ref{tdamt}$). Thus, $min_h-0.5 + \Sigma_{i=1, i \ne last}^d 2R_i$ $\leq$ $min_h+(d-1)/2 -\epsilon$ due to the restriction on total pebble weight (\ref{rtotal}). Then
$(d-1)/2 -\epsilon +0.5- \Sigma_{i=1, i \ne last}^d 2R_i$ is the maximum amount of pebble weight at $t_{last}$ on the root of $P_{last}$. It is the difference between the maximum pebble weight and the pebble weight elsewhere.

It is also the case that,\\
$(d-1)/2 -\epsilon +0.5- \Sigma_{i=1, i \ne last}^d 2R_i$ $\leq 0.5 + (d-1)/2 - \epsilon - (d-1-|G|) - \Sigma_{ i\in G} 2R_i$ (by \ref{GtonotG})\\
$= 0.5 + (d-1)/2 - \epsilon - (d-1) + |G| - \Sigma_{ i\in G} 2R_i$ \\
$= 0.5 - (d-1)/2 - \epsilon + |G| - \Sigma_{ i\in G} 2R_i$\\
We denote this quantity $Rmax$.
\begin{equation} Rmax = 0.5 - (d-1)/2 - \epsilon + |G| - \Sigma_{ i\in G} 2R_i \label{RMAX}\end{equation}

Thus $Rmax$ is an upper bound on the maximum amount of pebble weight at $t_{last}$ on the root of $P_{last}$. It is a measure dependent on the pebble weight maintained in the other subtrees.\\

\noindent
{\bf Case 2A :} $Rmax \geq 1$.

Note by assumption for {\bf Case 2A} and (\ref{RMAX}),

\begin{equation} -\Sigma_{i \in G} 2R_i \geq 0.5 + (d-1)/2 + \epsilon - |G| \label{rearge}\end{equation}

This will be used later in this {\bf Case}.\\

In this case we have not left enough pebble weight in the $P_i$, $i \ne last$.

Also in this case G is not the empty set. For contradiction, suppose it was. Then, 
$Rmax = 0.5 - (d-1)/2 - \epsilon + |G| - \Sigma_{ i\in G} 2R_i = 0.5 - (d-1)/2 - \epsilon \geq 1$. However, this is not possible since d $\geq$ 2 and $\epsilon \in (-0.5,0.5]$ (\ref{aeps}).

If 0 ${\leq}$ $R_i$ ${<}$ 0.5 we may apply {\bf Lemma \ref{flb1}} to the $i^{th}$ subtree at $t_{Pi-init}$ .

Thus, we apply {\bf Lemma \ref{flb1}} to $P_i$, $i$ $\in$ G, taking the initial time in the lemma to be $t_{Pi-init}$, taking E in the lemma to be ${R_i}$ and with $\tbith$ := $\tbss$ from the lemma. Then, $\tbith$ $>$ $\troots$, $w_{\pi}$($\tbith$)[$P_i$] $\geq$ $min_h + 0.5 - {R_i}$ and $w.w_\pi(t)[P_i] \geq 1 - 2{R_i}$ for t in [$\troots$, $\tbith$]. \\

We choose $\tbs$=min($\tbith$), $i$ $\in$ G. This is the first $\tbith$ which is reached in $\pi$. Let this $i=first$. Then we add $\Sigma_{i \in G, i \ne first} (1 - 2R_i)$ since we had yet to remove the pebble weight from the other $P_i$, $i \in G$,

\noindent
$sw_\pi(\tbs) \geq min_h+0.5-D_{first}+ \Sigma_{i \in G, i \ne first} (1 - 2R_i)$\\
$\geq min_h+0.5-2D_{first}+ \Sigma_{i \in G, i \ne first} (1 - 2R_i)$\\
$\geq min_h-0.5+ \Sigma_{i \in G} (1 - 2R_i)$\\
$= min_h-0.5+ |G| - \Sigma_{i \in G} 2R_i$\\
$\geq min_h - 0.5+ |G| + 0.5 + (d-1)/2 + \epsilon - |G|$ by (\ref{rearge}). \\
$= min_h + (d-1)/2 + \epsilon$\\
$= min_{h+1} + \epsilon$

Thus we exceed or match the minimum pebble weight allotted by the IH.\\

Also, we have $w.sw_\pi(t) \geq \Sigma_{i \in G} (1 - 2R_i)$ for $t$ in $[\troots, \tbs]$ since we have yet to remove the weight from any of the $P_{i}$. 

\noindent
$w.sw_\pi(t) \geq \Sigma_{i \in G} (1 - 2R_i)$\\
= $|G| - \Sigma_{i \in G} 2R_i$\\
$\geq |G| + 0.5 + (d-1)/2 + \epsilon - |G|$ by (\ref{rearge}).\\
=  $0.5 + (d-1)/2 + \epsilon$\\
$> 0.5 + \epsilon$ as required.\\

Thus in this case the IH is satisfied.\\

\noindent
{\bf Case 2B :} $Rmax < 1$ (\ref{RMAX})

Let $t_{D-before-last}$ be the last of the $t_{R_i}$ (see \ref{tdamt} and \ref{tdres}). If all $t_{R_i}$ are the initial time, choose any one arbitrarily as $t_{D-before-last}$. Let $P_{before-last}$ be the subtree associated with $t_{D-before-last}$ in the definition.

We wish to eventually apply {\bf Lemma \ref{flb2}} to $P_{last}$ for E = $Rmax$. To do this we take $t_{D-before-last}$ to be the initial time and $t_{last}$ to be the time $t_0$ in the lemma. To apply  {\bf Lemma \ref{flb2}}, we must show upper bounds on $b.sw_{\pi}$($t_{D-before-last}$)[$P_{last}$], $w_{\pi}(t)[P_{last}]$ for t in $[t_{D-before-last}, \troots]$, $b.rw_\pi(t_{last})[P_{last}]$ and we must show $Rmax$ $\in$ [0,1).\\

We first show $b.sw_{\pi}$($t_{D-before-last}$)[$P_{last}$] $\leq$ $0.5 + Rmax$. This is divided into cases.

\noindent
{\bf Case I :} $t_{D-before-last}$ was the initial time

If $t_{D-before-last}$ was the initial time, due to the restriction on initial black pebble weight (\ref{swinit}) and due to the pebble weight in the other subtrees ($\ref{tdamt}$),\\
$b.sw_{\pi}$($t_{D-before-last}$)[$P_{last}$] $\leq$ $1 - \epsilon - \Sigma_{i=1, i \ne last}^d 2R_i$\\
$\leq$ $1 - \epsilon - (d-1-|G|) - \Sigma_{i \in G} 2R_i$ (by \ref{GtonotG})\\
$\leq$ $0.5 + (d-1)/2  - \epsilon - (d-1-|G|) - \Sigma_{i \in G} 2R_i$ $(d \geq 2)$\\
$=$ $0.5 - (d-1)/2  - \epsilon + |G| - \Sigma_{i \in G} 2R_i$\\
$=$ $Rmax$ (by \ref{RMAX})\\
$\leq$ $0.5 + Rmax$ as required.

\medskip

\noindent
{\bf Case II :} $t_{D-before-last}$ was not the initial time

If $t_{D-before-last}$ was not the initial time, due to the restrictions on total pebble weight (\ref{rtotal}), the amount in $P_{before-last}$ ($\ref{tdamt}$) and the pebble weight in the other subtrees,\\
$b.sw_{\pi}$($t_{D-before-last}$)[$P_{last}$] $\leq$ $(d-1)/2 - \epsilon - D_{before-last} + 0.5 - \Sigma_{i=1, i \ne last , i \ne before-last}^d 2R_i$

\medskip

\noindent
{\bf Case IIA :}  $before-last$ is in G, therefore $D_{before-last} < 0.5$. 

There are $(d-1-|G|)$ other subtrees not in G since ${before-last}$ is in G. Thus if we continue from the above,\\
$b.sw_{\pi}$($t_{D-before-last}$)[$P_{last}$] $\leq$ $(d-1)/2 - \epsilon - D_{before-last} + 0.5 - (d-1-|G|) - \Sigma_{i \in G , i \ne before-last} 2R_i$ (similar to \ref{GtonotG}) \\
$\leq$ $(d-1)/2 - \epsilon - 2D_{before-last} + 1 - (d-1-|G|) - \Sigma_{i \in G , i \ne before-last} 2R_i$ \\
$=$ $(d-1)/2 - \epsilon + 1 - (d-1-|G|) - \Sigma_{i \in G} 2R_i$ \\
$=$ $1 - (d-1)/2  - \epsilon + |G| - \Sigma_{i \in G} 2R_i$\\
$=$ 0.5 + $Rmax$ (by \ref{RMAX}) as required.

\medskip

\noindent
{\bf Case IIB :} $before-last$ is not in G, therefore $D_{before-last} \geq 0.5$.

There are $(d-2-|G|)$ subtrees not in G other than ${before-last}$, since ${before-last}$ is not in G. Thus if we continue from what was described at the beginning of {\bf Case II},\\ 
$b.sw_{\pi}$($t_{D-before-last}$)[$P_{last}$] $\leq$ $(d-1)/2 - \epsilon - D_{before-last} + 0.5 - \Sigma_{i \in G} 2R_i - (d-2-|G|)$ (similar to \ref{GtonotG})\\
$\leq$ $(d-1)/2 - \epsilon - 0.5 + 0.5 - \Sigma_{i \in G} 2R_i - (d-2-|G|)$ \\
$=$ $(d-1)/2 - \epsilon - \Sigma_{i \in G} 2R_i - (d-2-|G|)$ \\
$=$ $(d-1)/2 - \epsilon - \Sigma_{i \in G} 2R_i - (d-2-|G|+1-1)$ \\
$=$ $(d-1)/2 - \epsilon - \Sigma_{i \in G} 2R_i - (d-1-|G|) + 1$ \\
$=$ $1 - (d-1)/2  - \epsilon + |G| - \Sigma_{i \in G} 2R_i$\\
$=$ 0.5 + $Rmax$ (by \ref{RMAX}) as required.\\ 
Thus in all cases the condition is met for the $b.sw_{\pi}$.\\

We next show $w_{\pi}(t)[P_{last}] \leq min_h-0.5+ Rmax$ for t in $[t_{D-before-last}, \troots]$. We use at most $w_{\pi}(t)[P_{last}]$ $\leq$ $min_h + (d - 1)/2 - \epsilon - \Sigma_{i=1, i \ne last}^d 2R_i$ for t in $[t_{D-before-last}, \troots]$ due to the pebble weight elsewhere ($\ref{tdres}$) and the restriction on total pebble weight before $\troots$ (\ref{rtotal}).\\
$w_{\pi}(t)[P_{last}]$ $\leq$ $min_h+(d-1)/2 - \epsilon - \Sigma_{i=1, i \ne last}^d 2R_i$\\
$\leq$ $min_h+(d-1)/2 - \epsilon - \Sigma_{i \in G}^d 2R_i - (d-1-|G|)$ (by \ref{GtonotG})\\
$=$ $min_h-(d-1)/2  - \epsilon + |G| - \Sigma_{i \in G} 2R_i$\\
$=$ $min_h-0.5+ Rmax$ (by \ref{RMAX}) as required.\\

We know by construction $b.rw_\pi(t_{last})[P_{last}] \leq Rmax$.\\
%At $t_{last}$, we use $sw_{\pi}(t_{last})[P_{last}] \geq min_h-0.5$ (\ref{swp2tlast}), while maintaining $\Sigma_{i=1, i \ne last}^d 2R_i$ ($\ref{tdres}$) in the other subtrees. Due to the restriction on total pebble weight (\ref{rtotal}) we have,\\
%$b.rw_{\pi}$($t_{last}$)[$P_{last}$] $\leq$ $0.5+(d-1)/2 - \epsilon - \Sigma_{i=1, i \ne last}^d 2R_i$  \\
%$\leq$ $0.5+(d-1)/2 - \epsilon - (d-1-|G|) - \Sigma_{i \in G} 2R_i$ (by \ref{GtonotG}) \\ 
%= $Rmax$\\

Finally we show $Rmax$ $\in$ [0,1). We use $sw_{\pi}(t_{last})[P_{last}] \geq min_h-0.5$ (\ref{swp2tlast}) while maintaining $\Sigma_{i=1, i \ne last}^d 2R_i$ in the other subtrees at time $t_{last}$ ($\ref{tdres}$). Thus $min_h-0.5 + \Sigma_{i=1, i \ne last}^d 2R_i$ $\leq$ $min_h+(d-1)/2 -\epsilon$ due to the restriction on total pebble weight (\ref{rtotal}). Then,\\
0 $\leq$ $(d-1)/2 -\epsilon +0.5- \Sigma_{i=1, i \ne last}^d 2R_i$ \\
$\leq 0.5 + (d-1)/2 - \epsilon - (d-1-|G|) - \Sigma_{ i\in G} 2R_i$  (by \ref{GtonotG})\\
= $Rmax$\\
Using this and the assumption, $Rmax$ $\in$ [0,1), as required.\\

Thus we have shown all the necessary conditions to apply {\bf Lemma \ref{flb2}} to $P_{last}$.

If 0 ${\leq}$ $R_i$ ${<}$ 0.5 we may apply {\bf Lemma \ref{flb1}} to the $i^{th}$ subtree at $t_{Pi-init}$.

Since $\troots$ occurs when $rw_\pi$($\troots$)[$P_{last}$]=1 and $rw_\pi$($\troots$)[$P_i$]=1, we apply {\bf Lemma \ref{flb2}} and {\bf Lemma \ref{flb1}}, respectively, taking $\troots$ as the time $\troot$ in the lemmas.\\

\noindent
We apply {\bf Lemma \ref{flb2}} to $P_{last}$ with $\tblast$ := $\tbss$ from the lemma. Then, $\tblast$ $>$ $\troots$, \\
$w_{\pi}$($\tblast$)[$P_{last}$] $\geq$ $min_h + 0.5- Rmax$\\
$= min_h+0.5- 0.5 + (d-1)/2 + \epsilon - |G| + \Sigma_{ i\in G} 2R_i$ (by \ref{RMAX})\\
$= min_h + (d-1)/2 + \epsilon - |G| + \Sigma_{ i \in G} 2R_i$ \\
and \\
$w.w_\pi(t)[P_{last}] \geq 1 - Rmax$\\
$= 1- 0.5 + (d-1)/2 + \epsilon -|G| + \Sigma_{ i\in G} 2R_i$ (by \ref{RMAX})\\
$= 0.5 + (d-1)/2 + \epsilon - |G| + \Sigma_{ i \in G} 2R_i$ for t in [$\troots$, $\tblast$].\\

\noindent
We apply {\bf Lemma \ref{flb1}} to $P_i$, $i$ $\in$ G, taking the initial time in the lemma to be $t_{Pi-init}$, taking E in the lemma to be ${R_i}$ and with $\tbith$ := $\tbss$ from the lemma. We may do this since $b.sw_{\pi}(0) \leq 2{R_i} \leq 0.5+{R_i}$ and $b.rw_{\pi}(0) \leq 2{R_i}$. Then, $\tbith$ $>$ $\troots$, $w_{\pi}$($\tbith$)[$P_i$] $\geq$ $min_h + 0.5 - {R_i}$ and $w.w_\pi(t)[P_i] \geq 1 - 2{R_i}$ for t in [$\troots$, $\tbith$].\\

\noindent
We choose $\tbs$=min($\tblast$, $\tbith$) for i $\in$ G.\\

\noindent
{\bf Case 2B-1 :} $\tbs$ = $\tblast$. Then,\\
$sw_\pi(\tbs) \geq min_h + (d-1)/2 + \epsilon - |G| + \Sigma_{ i \in G} 2R_i+ \Sigma_{i \in G} (1-2R_i)$\\
$= min_h + (d-1)/2 + \epsilon - |G| + \Sigma_{ i \in G} 2R_i + |G| - \Sigma_{i \in G} 2R_i$\\
$= min_h + (d-1)/2 + \epsilon$\\
$= min_{h+1} + \epsilon$\\
Where we add the pebble weight in the $P_i$s since we had yet to reach the $\tbith$. Thus we exceed or match the minimum pebble weight allotted by the IH.

Also, we have white pebble weight as follows between [$\troots$,$\tbs$],\\
$w.sw_\pi(t) \geq 0.5 + (d-1)/2 + \epsilon - |G| + \Sigma_{ i \in G} 2R_i+ \Sigma_{i \in G} (1 - 2{R_i})$\\
= $0.5 + (d-1)/2 + \epsilon - |G| + \Sigma_{ i \in G} 2R_i+ |G| - \Sigma_{i \in G} 2{R_i}$\\
= $0.5 + (d-1)/2 + \epsilon$\\
$\geq 0.5 + \epsilon$ as required.\\

Thus the IH is satisfied in this case.\\

\noindent
{\bf Case 2B-2 :} 
$\tbs$ = $\tbith$, $i \ne last$. 

We let this $i = first$. Then,\\
$sw_\pi(\tbs) \geq min_h+0.5-D_{first}+ 0.5 + (d-1)/2 + \epsilon - |G| + \Sigma_{ i \in G} 2R_i+  \Sigma_{i \in G, i \ne first} (1 - 2R_i)\\
\geq min_h+1-2D_{first}+ (d-1)/2 + \epsilon - |G| + \Sigma_{ i \in G} 2R_i +  \Sigma_{i \in G, i \ne first} (1 - 2R_i)$\\
$=min_h + (d-1)/2 + \epsilon - |G| + \Sigma_{ i \in G} 2R_i +  \Sigma_{i \in G} (1 - 2R_i)$\\
$= min_h + (d-1)/2 + \epsilon - |G| + \Sigma_{ i \in G} 2R_i + |G| - \Sigma_{i \in G} 2R_i$\\
$= min_h + (d-1)/2 + \epsilon$\\
$= min_{h+1} + \epsilon$\\
This matches the lower bounds specified in the IH.

As in {\bf Case 2B-1}, we have the same amount of white pebble weight until this time. Thus the IH is satisfied in this case.\\\\

Thus the IH holds in all cases. Consequently the main theorem holds as well.

\newpage 
\section{Conclusion} 

We have presented a proof of an open problem given in \cite{c:pebjournal}. Fractional pebbles allow for many pebbling strategies. To accommodate for this, we used a $shifting$ argument to build a direct proof. Many open problems remain related to the fractional pebbling game.\\

Branching programs were briefly introduced in the introduction (Section 1). They are nonuniform models of Turing machines. Showing that non-deterministic branching programs require a superpolynomial number of states for a problem in P would separate NL from P. 

\cite{c:pebjournal} proposed the $tree$ $evaluation$ $problem$ as a mean of separating NL from P. The $tree$ $evaluation$ $problem$ is similar to the pebbling game except values are attached to each leaf node and functions are attached to each non-leaf node. The value of a node is determined by the value of its function evaluated at the value of its children. The goal is then to determine the value of the root node. 

One step towards separating NL from P is to show a superpolynomial lower bound on the number of states for a restricted class of branching programs. A thrifty branching program for the $tree$ $evaluation$ $problem$ must query the value of the functions only at the correct value of the children. The thrifty hypothesis states that thrifty branching programs are optimal among all branching programs.

\cite{c:pebjournal}, under the thrifty hypothesis, showed that deterministic branching programs solving the $tree$ $evaluation$ $problem$ required a superpolynomial number of states that would separate L from P. This followed from a proof similar to the one in Section 3.2. Thus we propose the following as an open problem :

%it is the hope that a direct proof of lower bounds for the $fractional $ $pebbling$ $game$, such as the one we have given, could be adapted to a proof of the following for non-deterministic branching programs :\\

\begin{open}
Adapt the proof of the Main Theorem to get lower bounds for non-deterministic thrifty branching programs solving the $tree$ $evaluation$ $problem$.
\end{open}

Showing this would separate NL from P under the thrifty hypothesis. To show their original result, \cite{c:pebjournal} used a non-inductive proof. It seems difficult to instead use an inductive proof, thus the following would be interesting :

\begin{open}
Provide an alternative proof, using induction, that under the thrifty hypothesis, deterministic thrifty branching programs solving the $tree$ $evaluation$ $problem$ require a superpolynomial number of states which would separate L from P. 
\end{open}

If this could be done without the thrifty hypothesis it would be an even more important result. Similarly, showing that the thrifty hypothesis held or did not is an important open problem.\\

Klawe showed the lower bound for the whole $black$-$white$ $pebbling$ $game$ for the pyramid graphs \cite{k:bwpyr}. The advantage of the pyramid graphs is that the number of nodes is polynomial in the height of the tree. Thus for various application of the pebbling game, it is possible that lower bounds for the pyramid graphs could result in better bounds. We thus suggest the following open problem :

\begin{open}
Show upper bounds and lower bounds for the $fractional $ $pebbling$ $game$ on pyramid graphs.
\end{open}

\newpage

\section*{Acknowledgements} 
\addcontentsline{toc}{section}{Acknowledgements}

I would like to thank my advisors Stephen A. Cook and Toniann Pitassi. They originally suggested
this problem and provided many of the key arguments and approaches found in the proof. My
experience would not have been as successful nor as pleasant without the weekly meetings and
constant support they provided.

\smallskip

I would like to thank the authors of \cite{c:pebjournal}, whose work was fascinating and formed the
basis for mine. Particularly, I would like to thank Dustin Wehr who was very helpful at various
points during my research.

\smallskip

Finally, I must thank my family and friends, whose love and support has carried me to and
through this wonderful experience.

\hfill \textit{Frank Vanderzwet}

\newpage

\addcontentsline{toc}{section}{References}
\bibliographystyle{cell}
\bibliography{frank}

\end{document}